\documentclass[copyright,creativecommons]{eptcs}
\providecommand{\event}{TARK 2019}
\usepackage{breakurl}

\usepackage{colonequals}
\usepackage{etoolbox}
\usepackage{amssymb}
\usepackage{amsmath}
\usepackage{amsthm}

\usepackage{textcomp}

\usepackage{mathrsfs}
\usepackage{paralist}
\usepackage{tikz}
\usepackage{tikz-cd}

\newtheorem{theorem}{Theorem}
\newtheorem{lemma}[theorem]{Lemma}

\theoremstyle{definition} 
\newtheorem{definition}[theorem]{Definition}
\theoremstyle{remark}

\newcommand{\bbbt}{\mathbb{T}}
\newcommand{\bbbn}{\mathbb{N}}
\newcommand{\ce}{\colonequals}
\newcommand{\cce}{\coloncolonequals}
\newcommand{\notni}{\not\owns}

\newcommand{\gevents}[1][]{\mathit{GEvents}_{#1}}
\newcommand{\bevents}[1][]{\mathit{BEvents}_{#1}}
\newcommand{\fevents}[1][]{\mathit{FEvents}_{#1}}
\newcommand{\sevents}[1][]{\mathit{SysEvents}_{#1}}
\newcommand{\gtrueevents}[1][]{\overline{\mathit{GEvents}}_{#1}}
\newcommand{\gtrueactions}[1][]{\overline{\mathit{GActions}}_{#1}}
\newcommand{\gactions}[1][]{\gtrueactions[#1]}
\newcommand{\gthings}[1][]{\ghaps[#1]}
\newcommand{\gtruethings}[1][]{\overline{\ghaps[#1]}}

\newcommand{\events}[1][]{{\mathit{Events}}_{#1}}
\newcommand{\actions}[1][]{{\mathit{Actions}}_{#1}}
\newcommand{\haps}[1][]{{\mathit{Haps}}_{#1}}
\newcommand{\ghaps}[1][]{{\mathit{GHaps}}_{#1}}

\newcommand{\truethings}[1][]{\haps[#1]}
\newcommand{\agents}{\mathcal{A}}

\newcommand{\msgs}[1][]{\Msgs[#1]}
\newcommand{\Msgs}[1][]{\mathit{Msgs}^{#1}}

\newcommand{\glob}[2]{{\mathop{\mathit{global}}\left(\ifstrempty{#1}{\dots}{#1},\ifstrempty{#2}{\dots}{#2}\right)}}

\newcommand{\gsend}[4]{\mathop{\mathit{gsend}}(%
    \ifstrempty{#1}{\dots}{#1},%
    \ifstrempty{#2}{\dots}{#2},%
    \ifstrempty{#3}{\dots}{#3},%
    \ifstrempty{#4}{\dots}{#4}%
)}
\newcommand{\send}[2]{\mathop{\mathit{send}}(%
    \ifstrempty{#1}{\dots}{#1},%
    \ifstrempty{#2}{\dots}{#2}%
)}

\newcommand{\grecv}[4]{\mathop{\mathit{grecv}}(%
    \ifstrempty{#1}{\dots}{#1},%
    \ifstrempty{#2}{\dots}{#2},%
    \ifstrempty{#3}{\dots}{#3},%
    \ifstrempty{#4}{\dots}{#4}%
)}
\newcommand{\recv}[2]{\mathop{\mathit{recv}}(%
    \ifstrempty{#1}{\dots}{#1},%
    \ifstrempty{#2}{\dots}{#2}%
)}

\newcommand{\fakeof}[2]{{\mathop{\mathit{fake}}\left(#1,#2\right)}}
\newcommand{\localstates}[2][]{%
    {\mathscr{L}^{#1}_{#2}}%
}
\newcommand{\globalstates}{\mathscr{G}}
\newcommand{\globalinitialstates}[1][]{\mathscr{G}^{#1}(0)}

\newcommand{\failof}[1]{{\mathop{\mathit{fail}}\left(#1\right)}}
\newcommand{\sleep}[1]{{\mathop{\mathit{sleep}}\left(#1\right)}}
\newcommand{\hib}[1]{{\mathop{\mathit{hibernate}}\left(#1\right)}}

\newcommand{\tick}{\textbf{noop}}

\newcommand{\mistakefor}[2]{#1\mapsto#2}

\newcommand{\correct}[1]{correct_{#1}}

\newcommand{\fake}[2][]{\mathit{fake}_{#1}\left(#2\right)}

\newcommand{\occurred}[2][]{\mathit{occurred}_{#1}(#2)}
\newcommand{\trueoccurred}[2][]{\overline{\mathit{occurred}}_{#1}(#2)}

\newcommand{\multipede}[2]{\mbox{multipede}^{#1}_{#2}}

\newcommand{\envprotocol}[1]{P_{\epsilon}\ifstrempty{#1}{}{\left(#1\right)}}
\newcommand{\agprotocol}[2]{{P_{#1}\ifstrempty{#2}{}{\left(#2\right)}}}
\newcommand{\joinprotocol}[1]{P\ifstrempty{#1}{}{\left(#1\right)}}

\newcommand{\failed}[3][]{{Failed_{#1}\left(#2\ifstrempty{#3}{}{,#3}\right)}}

\newcommand{\adversary}{adversary}

\newcommand{\alphae}[3][]{{%
    \alpha_{\epsilon_{#1}}^{#3}\ifstrempty{#2}{}{\left({#2}\right)}%
}}

\newcommand{\alphaag}[3]{{%
    \alpha_{#1}^{#3}\ifstrempty{#2}{}{\left(#2\right)}%
}}

\newcommand{\betae}[3][]{{%
    \beta_{\epsilon_{#1}}^{#3}\ifstrempty{#2}{}{\left({#2}\right)}%
}}

\newcommand{\betaag}[3]{{%
    \beta_{#1}^{#3}\ifstrempty{#2}{}{\left(#2\right)}%
}}

\newcommand{\filtere}[3][]{filter^{#1}_{\epsilon}\ifstrempty{#2}{}{\left(#2,#3\right)}}
\newcommand{\filterag}[4][]{filter^{#1}_{#2}\ifstrempty{#3}{}{\left(#3,#4\right)}}
\newcommand{\update}[2]{update\ifstrempty{#1}{}{\left(#1,#2\right)}}
\newcommand{\updatee}[2]{update_{\epsilon}\ifstrempty{#1}{}{\left(#1,#2\right)}}
\newcommand{\updateag}[4]{update_{#1}\ifstrempty{#2}{}{\left(#2,#3,#4\right)}}

\newcommand{\sigmaof}[1]{\sigma\ifstrempty{#1}{}{\bigl(#1\bigr)}}

\newcommand{\points}{{\agents\times\bbbt}}

\newcommand{\run}[3][]{r#1_{#2}\left(#3\right)}

\newcommand{\transitionExt}[3][{\envprotocol{},\joinprotocol{}}]{\tau_{\ifstrempty{#2}{}{#2,}#1}\ifstrempty{#3}{}{\left({#3}\right)}}

\newcommand{\system}[1]{{R^{#1}}}

\newcommand{\FS}{\mathit{FS}}

\newcommand{\K}[2]{K_{#1}\ifstrempty{#2}{}{#2}}

\newcommand{\adj}{\mathit{adj}}

\newcommand{\SystemBuildFrom}[3]{R\left({#1},{#2},{#3}\right)}

\newcommand{\newfreezerule}[1]{\mathit{CFreeze}\ifstrempty{#1}{}{\left({#1}\right)}}

\newcommand{\newchatterrule}[4]{{#3}\text{-}\mathit{Focus}\mathstrut_{#1}^{{#2}}\ifstrempty{#4}{}{\left({#4}\right)}}
\newcommand{\newechorule}[3]{\mathit{FakeEcho}_{#1}^{{#2}}\ifstrempty{#3}{}{\left({#3}\right)}}

\renewcommand{\vDash}{\models}
\renewcommand{\nvDash}{\not\models}

\newcommand{\A}{\mathfrak{a}}
\newcommand{\E}{\mathfrak{e}}

\newcommand{\I}{\mathcal{I}}
\newcommand{\prop}{\mathit{Prop}}

\newcommand{\edgLoc}{\rightarrow_{\mathit{l}}}
\newcommand{\edgMsg}[1]{{\rightarrow_{c}^{#1}}}
\newcommand{\edgC}[1]{{\rightarrow^{#1}}}

\newcommand{\pastpath}[4][]{{{#3}\mathrel{\rightsquigarrow^{#2}_{#1}}{#4}}}

\newcommand{\truecausalcone}[2]{\blacktriangleright_{#2}^{#1}}

\newcommand{\faultbuffer}[2]{\mathord{\rangle\!\rangle^{#1}_{#2}}}
\newcommand{\silentmasses}[2]{\mathord{\cdot\!\rangle^{#1}_{#2}}}
\newcommand{\havevoice}[2]{\mathord{\rangle\!\rangle\!\!\!\blacktriangleright\!\!^{#1}_{#2}}}

\title{Causality and Epistemic Reasoning \\in Byzantine Multi-Agent Systems}
\def\titlerunning{Causality and Epistemic Reasoning in Byzantine Multi-Agent Systems}

\author{Roman Kuznets\thanks{Supported by the Austrian Science Fund (FWF) projects RiSE/SHiNE 
(S11405) and ADynNet (P28182).} 
\institute{Embedded Computing Systems\\ TU Wien\\ Vienna, Austria}
\email{roman@logic.at}
\and
Laurent Prosperi
\institute{ENS Paris-Saclay\\ Cachan, France}
\email{laurent.prosperi@ens-cachan.fr}
\and
Ulrich Schmid
\institute{Embedded Computing Systems\\ TU Wien\\ Vienna, Austria}
\email{s@ecs.tuwien.ac.at}
\and
Krisztina Fruzsa\thanks{PhD student in the FWF doctoral program LogiCS (W1255).}
\institute{Embedded Computing Systems\\ TU Wien\\ Vienna, Austria}
\email{krisztina.fruzsa@tuwien.ac.at}
}
\def\authorrunning{R. Kuznets, L. Prosperi, U. Schmid \& K. Fruzsa}

\begin{document}
\maketitle

\begin{abstract}
Causality is an important concept both for proving impossibility results and for synthesizing efficient protocols in distributed computing. For asynchronous agents communicating over unreliable channels, causality is well studied and understood. This understanding, however, relies heavily on the assumption that agents themselves are correct and reliable. We provide the first epistemic analysis of causality in the presence of byzantine agents, i.e., agents that can deviate from their protocol and, thus, cannot be relied upon. Using our new framework for epistemic reasoning in fault-tolerant multi-agent systems, we determine the byzantine analog of the causal cone and describe a communication structure, which we call a multipede, necessary for verifying preconditions for actions in this setting.\looseness=-1
\end{abstract}

\section{Introduction}
Reasoning about knowledge has been a valuable tool for analyzing distributed systems for decades~\cite{bookof4,HM90}, and has provided a number of fundamental insights. As crisply formulated by Moses~\cite{Mos15TARK} in the form of the \emph{Knowledge of Preconditions Principle}, a precondition for action must be known  in order to be actionable. In a distributed environment, where  agents only communicate  by exchanging messages, an agent can only learn about events happening to other agents via messages (or sometimes the lack thereof~\cite{GM18:PODC}). 
\looseness=-1

In \emph{asynchronous systems}, where the absence of communication is indistinguishable from delayed communication,  agents can only rely on messages they receive. Lamport's seminal definition of the \emph{happened-before} relation~\cite{Lam78} establishes the causal structure for asynchronous agents in the agent--time graph describing a run of a system. This structure is often referred to as a \emph{causal cone}, whereby causal links are either time transitions from past to future for one agent or messages from one agent to another. As demonstrated by Chandy and Misra~\cite{ChaMis86DC},  the behavior of an asynchronous agent can only be affected by events from within its causal cone. 

The standard way of showing that an agent does not know  of an event 
is to modify a given run by removing the event in question in such a way that the agent cannot detect the change.  By Hintikka's definition of knowledge~\cite{Hin62}, the agent thinks it possible that  the event has not occurred  and, hence, does not know of  the event to have occurred. Chandy and Misra's result shows that in order for agent~$i$ to learn of an event happening to another agent~$j$, there must exist a chain of successfully delivered messages leading from the moment of agent~$j$ observing the event to some past or present state of agent~$i$. 
This observation remains valid in asynchronous distributed systems where messages could be lost and/or where agents may stop operating (i.e.,~crash)~\cite{dwork1990knowledge,halpern2001characterization,MT88}.
\looseness=-1

In \emph{synchronous systems}, if message delays are upper-bounded, agents can also learn from the absence
of communication (com\-mun\-ic\-ation-by-time). As shown in~\cite{ben2014beyond},  Lamport's happened-before relation must then be augmented by causal links indicating no communication within the message delay upper bound to also capture causality induced via communication-by-time, leading to the so-called \emph{syncausality relation}. Its utility has been demonstrated using the \emph{ordered response problem}, where  agents must perform a sequence of actions in a given order: both the necessary and sufficient knowledge and a necessary and sufficient communication structure (called a \emph{centipede}) have been determined in~\cite{ben2014beyond}. It is important to note, however, that syncausality works only in fault-free distributed systems with reliable 
communication. Although  it has recently been shown in~\cite{GM18:PODC} that \emph{silent choirs} 
are a way to extend it to distributed systems where agents may crash, the idea does not generalize to less benign faults.
\looseness=-1

Unfortunately, all the above ways of capturing causality and the resulting simplicity of determining the causal cone completely break down
if agents may be \emph{byzantine} 
faulty~\cite{lamport1982byzantine}.  Byzantine faulty agents may behave arbitrarily, in particular, need not adhere to their protocol and may, hence, send arbitrary messages.  It is common to limit the maximum number of agents that ever act byzantine in a distributed system by some number~$f$, which is typically much smaller than the total number~$n$ of agents. 
Prompted by the ever growing number of faulty hardware and software having real-world negative, sometimes life-critical, consequences, capturing causality and providing ways for determining the causal cone in byzantine fault-tolerant distributed systems is both an important and scientifically challenging task. To the best of our knowledge, this challenge has not been addressed in the literature before.\footnote{Despite having ``Byzantine'' in the title, \cite{dwork1990knowledge,halpern2001characterization}~only address benign faults (crashes, send/receive omissions of messages).}

In a nutshell, for $f>0$, the problem of capturing causality becomes complicated by the fact that a simple causal chain of messages is no longer sufficient: a single byzantine agent in the chain could manufacture ``evidence'' for anything, both false negatives and false positives. And indeed, obvious generalizations of message chains do not work. For example, it is a folklore result that, in the case of direct communication, at least \mbox{$f+1$}~confirmations are necessary because $f$~of them could be false. When information is transmitted along arbitrary, possibly branching and intersecting chains of messages, the situation is even more complex and defies simplistic direct analysis. In particular, as shown by the counterexample 
in~\cite[Fig.~1]{Mengernomore}, one cannot rely on Menger's Theorem~\cite{diestel2000graphentheory} for separating nodes in the two-dimensional agent--time graph.
\looseness=-1

\textbf{Major contributions:} In this paper, we generalize the causality structure of asynchronous distributed systems described above to multi-agent systems involving byzantine faulty agents. Relying on our novel  byzantine runs-and-systems framework~\cite{FroCoS19} (described in full detail in~\cite{PKS19:TR}), we utilize some generic epistemic analysis results for determining the shape of the byzantine analog of Lamport's causal cone. Since \emph{knowledge} of an event is too strong a precondition in the presence of byzantine agents, it has to be relaxed to something more akin to \emph{belief} relative to correctness~\cite{MosSho93AI}, for which we coined the term \emph{hope}. We show that hope can only be achieved via a causal message chain that passes solely through correct agents (more precisely, through agents still correct while sending the respective messages). 
While the result looks natural enough, its formal proof is quite involved technically and paints an instructive picture of how byzantine agents can affect the information flow. We also establish a necessary condition for  detecting an event, and a corresponding communication structure (called a \emph{multipede}), which is severely complicated by the fact that the reliable causal cones of indistinguishable runs may be different.
\looseness=-1

\textbf{Paper organization:} In Sect.~\ref{sec:basics}, we succinctly introduce the features 
of our byzantine
runs-and-systems framework~\cite{PKS19:TR} and state some generic theorems and lemmas needed for proving the results of the paper. In Sect.~\ref{sect:run_mod},
we describe the  mechanism of run modifications, which are used to remove events an agent should not know about from a run, without the agent noticing. Our characterization of the byzantine causal cone is provided in Sect.~\ref{sect:ByzCausalCone}, the
necessary conditions for establishing hope for an occurrence
of an event and the underlying multipede structure can be found in 
Sect.~\ref{sect:PrecActions}. Some conclusions in Sect.~\ref{sect:conclusions} round-off the paper. 
\looseness=-1

\section{Runs-and-Systems Framework for Byzantine Agents}
\label{sec:basics}
First, we describe the modifications of the runs-and-systems framework~\cite{bookof4} necessary to account for byzantine behavior. 
To prevent wasting space on multiple definition environments, we give the following series of formal definitions as ordinary text marking defined objects by italics;
consult~\cite{PKS19:TR} for the same  definitions in fully spelled-out format.
As a further space-saving measure, instead of repeating every time ``actions and/or events,'' we use \emph{haps}\footnote{Cf. ``Till I know 'tis done, Howe'er my haps, my joys were ne'er begun.'' W.~Shakespeare, \emph{Hamlet}, Act~IV, Scene~3.\looseness=-1} as  a general term referring to either actions or events. \looseness=-1

The goal of all these definitions is to formally describe a system where \emph{asynchronous agents} $1,\dots,n$ perform actions according to their protocols, observe events, and exchange messages within an environment represented as a special agent~$\epsilon$. Unlike the environment, agents only have limited local information, in particular, being asynchronous, do not have access to the global clock. No assumptions apart from liveness are made about the communication. Messages can be lost, arbitrarily delayed, and/or delivered in the wrong order. This part of the system is a fairly standard asynchronous system with unreliable communication. The novelty is that the environment may additionally cause at most $f$~agents to become faulty \emph{in arbitrary ways}. A faulty agent can perform any of its actions irrespective of its protocol and observe events that did not happen, e.g.,~receive unsent or corrupted messages. It can also have false memories about actions it has performed. At the same time, much like the global clock, such malfunctions are not directly visible to an agent, especially when it mistakenly thinks it acted correctly.\looseness=-1

We fix a finite set $\agents=\{1,\dots,n\}$ of agents. Agent $i\in\agents$ can perform \emph{actions} $a \in \actions[i]$, e.g.,~send \emph{messages}, and witness \emph{events} $e\in \events[i]$ such as message delivery. We denote
 $\haps[i] \ce \actions[i]\sqcup\events[i]$.
The action of sending a copy numbered~$k$ of a message  $\mu \in \msgs$ to an agent $j\in\agents$ is denoted $\send{j}{\mu_k}$, whereas a receipt of such a message from  $i\in\agents$ is recorded locally as $\recv{i}{\mu}$.\footnote{Thus, it is possible to send several copies of the same message in the same round. If one or more of such copies are received in the same round,  however, the recipient does not know which copy it has received, nor that there have been multiple copies.\looseness=-1} \looseness=-1

Agent~$i$ records actions from $\actions[i]$ and observes events from $\events[i]$ without dividing them into correct and faulty. The environment~$\epsilon$, on the contrary, always knows if the agent acted correctly or was forced into byzantine behavior. Hence, the syntactic representations of each hap for agents (local view) and for the environment (global view) must differ, with the latter containing more information. In particular, the global view syntactically distinguishes correct haps from their byzantine  counterparts.  
While there is no way for an agent to distinguish a real event from its byzantine duplicate, it can analyze its recorded actions and compare them with its protocol. Sometimes, this information might be sufficient for the agent to detect its own malfunctions.\looseness=-1

All of $\actions \ce \bigcup_{i\in \agents}\actions[i]$, $\events \ce  \bigcup_{i\in \agents}\events[i]$, and  $\haps \ce \actions \sqcup \events$ represent the local view of haps.  
All  haps taking place after a \emph{timestamp}~$t\in\bbbt\ce\bbbn$ and no later than~\mbox{$t+1$} are grouped into a \emph{round} denoted~\mbox{$t+{}$\textonehalf{}} and are treated as  happening simultaneously. 
To model \emph{asynchronous agents}, we exclude these system timestamps  from the local format of $\haps$. 
At the same time, the environment~$\epsilon$ incorporates the current time\-stamp~$t$ into the global format of every correct action $a\in\actions[i]$, as initiated by agent~$i$ in the local format, via a one-to-one function $\glob{i,t}{a}$. Time\-stamps are especially crucial for proper message processing with $
\glob{i,t}{\send{j}{\mu_k}} \ce \gsend{i}{j}{\mu}{id(i,j,\mu,k,t)}
$ 
for some one-to-one function $id \colon \agents \times \agents \times \msgs \times \bbbn \times \bbbt \to \bbbn$ that assigns each sent message a unique \emph{global message identifier}~(GMI). 
We chose not to model agent-to-agent channels explicitly. With all messages effectively sent through one system-wide channel, these~GMIs  are needed to ensure the causality of message delivery, i.e.,~that only sent messages can be delivered correctly. The sets
$\gtrueactions[i] \ce \{\glob{i,t}{a} \mid t\in \bbbt, a \in \actions[i]\}$  of all possible correct actions for each agent in global format are pairwise disjoint due to the injectivity of $\mathit{global}$.  We set $\gtrueactions \ce \bigsqcup_{i\in\agents}\gtrueactions[i]$. \looseness=-1

Unlike correct actions, correct events witnessed by agents  are generated by the environment~$\epsilon$ and, hence, can be assumed to be produced already in the global format $\gtrueevents[i]$. 
We define $\gtrueevents \ce \bigsqcup_{i\in \agents} \gtrueevents[i]$ assuming them to be pairwise disjoint, and $\gtruethings = \gtrueevents \sqcup \gtrueactions$.
We do not consider the possibility of the environment violating its protocol, which is meant to model the fundamental physical laws of the system. 
Thus, all events that can happen are considered correct. 
A byzantine event is, thus, a subjective notion. It is an event that was perceived by an agent despite not taking place. In other words, each correct event $E \in \gtrueevents[i]$ has a faulty counterpart $\fakeof{i}{E}$, and  agent~$i$ cannot distinguish  the two. 
An important type of correct global events of agent~$j$ is  the delivery  $\grecv{j}{i}{\mu}{id}\in \gtrueevents[j]$ of message~$\mu$ with GMI $id \in \bbbn$ sent by agent~$i$. 
Note that the~GMI, which is used in by the global format to ensure causality, must be removed before the delivery is recorded by the agent in the local format because GMIs~contain the time of sending, which should not be accessible to agents.
To strip this information before updating local histories, we employ a function 
 $
 \mathit{local} \colon  \gtruethings \to  \haps
 $
 converting \emph{correct} haps from the global into the local format 
 in such a way that for actions $\mathit{local}$ reverses $\mathit{global}$, i.e.,~$\mathit{local}\bigl(\glob{i,t}{a}\bigr) \ce a$. For message deliveries,   $\mathit{local}{\bigl(\grecv{j}{i}{\mu}{id}\bigr)} \ce \recv{i}{\mu}$, i.e.,~agent~$j$ only knows that it received message~$\mu$ from agent~$i$.
It is, thus, possible for two distinct correct global events, e.g., $\grecv{j}{i}{\mu}{id}$ and $\grecv{j}{i}{\mu}{id'}$,  representing the delivery of different copies of the same message~$\mu$, possibly sent by~$i$ at different times, to be recorded by~$j$ the same way, as $\recv{i}{\mu}$. \looseness=-1

Therefore, correct actions are initiated by agents in the local format and translated into the global format by the environment. Correct and  byzantine  events  are initiated by the environment in the global format and translated into the local format before being recorded by agents.\footnote{This has already been described for correct events. A byzantine event is recorded the same way as its correct counterpart.} We will now turn our attention
to byzantine actions.

While a faulty event is purely an error of perception, actions can be faulty in another way: they can violate the protocol. 
The crucial question is: who should be responsible for such violations? 
With agents' actions governed by their protocols while everything else is up to  the environment, it seems that errors, especially unintended errors, should be the environment's responsibility. 
A malfunctioning agent tries to follow its protocol but fails for reasons outside of its control, i.e., due to environment interference. 
A malicious agent tries to hide its true intentions from other agents by pretending to follow its expected protocol and, thus, can also be modeled via environment interference. 
Thus, we model faulty actions as byzantine events of the form $\fakeof{i}{\mistakefor{A}{A'}}$ where $A, A' \in \gtrueactions[i] \sqcup\{\tick\}$ for a special \emph{non-action} $\tick$ in global format. Here $A$~is the action (or, in case of $\tick$, inaction) performed, while $A'$~represents the action (inaction) perceived instead by the agent. More precisely, the agent  either records $a' = \mathit{local}(A')\in \events[i]$  if $A' \in \gtrueevents[i]$ or has no record of this byzantine action  if $A' = \tick$.
The byzantine inaction $\failof{i} \ce \fakeof{i}{\mistakefor{\tick}{\tick}}$ is used to make agent $i$ faulty without performing any actions and without leaving a record in $i$'s~local history. 
The set of all $i$'s~byzantine events, corresponding to both faulty events and faulty actions, is denoted $\bevents[i]$, with $\bevents \ce \bigsqcup_{i\in\agents} \bevents[i]$. \looseness=-1

To prevent our asynchronous agents from  inferring the global clock by counting rounds, we make waking up for a round contingent on the environment issuing  a special \emph{system event} $go(i)$ for the agent in question. 
Agent~$i$'s local view of the system immediately after round~\mbox{$t+{}$\textonehalf}, referred to as (\emph{process-time} or \emph{agent-time}) \emph{node}~$(i,t+1)$, is recorded in $i$'s~\emph{local state} $r_i(t+1)$, also called $i$'s~\emph{local history}. 
Nodes~$(i,0)$ correspond to \emph{initial local states} $r_i(0)\in\Sigma_i$, with
$\globalinitialstates \ce \prod_{i\in \agents}\Sigma_i$. 
If a round contains neither  $go(i)$ nor any event to be recorded in  $i$'s~local history, then the said history $r_i(t+1)=r_i(t)$ remains unchanged, denying the agent the knowledge of the round just passed. 
Otherwise, $r_i(t+1) = X \colon r_i(t)$, for   $X \subseteq \haps[i]$, the set   of  all 
actions   and events   perceived by~$i$  in round~\mbox{$t+{}$\textonehalf}, where ${}\colon{}$~stands for concatenation.  
The exact definition will be given via the $\updateag{i}{}{}{}$ function, to be described shortly. Thus, 
the local history~$r_i(t)$ is  a list of all haps as perceived by~$i$ in rounds it was \emph{active} in. The set of all local states of~$i$ is~$\localstates{i}$.\looseness=-1

While not necessary for asynchronous agents, for future backwards compatibility, we add more system events for each agent, to serve as faulty counterparts to $go(i)$. 
Commands $\sleep{i}$ and $\hib{i}$ signify a failure to activate the agent's protocol and differ in that the former enforces waking up the agent (and thus recording time) notwithstanding. 
These commands will be used, e.g.,~for synchronous  systems. 
None of the \emph{system events} $\sevents[i] \ce \{go(i), \sleep{i},\hib{i}\}$ is directly detectable by agents.\looseness=-1

To summarize, $\gevents[i] \ce \gtrueevents[i] \sqcup \bevents[i] \sqcup \sevents[i]$ with $\gevents\ce \bigsqcup_{i\in\agents} \gevents[i]$ and $\ghaps \ce \gevents \sqcup \gactions$. Throughout the paper, horizontal bars signify phenomena that are correct. Note that   the absence of this bar means the absence of a claim of correctness. It does not necessarily imply a fault. Later, this would also apply to formulas, e.g.,~$\trueoccurred[i]{e}$ demands a correct occurrence of an event~$e$ whereas  $\occurred[i]{e}$ is satisfied by either correct or faulty occurrence.\looseness=-1

We now turn to the description of runs and protocols for our byzantine-prone asynchronous agents. A \emph{run}~$r$ is a sequence of \emph{global states} $r(t) = (r_{\epsilon}(t),r_1(t),\dots,r_n(t))$ of the whole system consisting of the \emph{state $r_\epsilon(t)$ of the environment} and local states $r_i(t)$ of every agent. We already discussed the composition of local histories. 
Similarly, the \emph{environment's history} $r_\epsilon(t)$ is a list of all haps that happened, this time faithfully recorded in the global format. Accordingly, 
$r_\epsilon(t+1) = X \colon r_\epsilon(t)$ for the set $X \subseteq \ghaps$ of all haps from round~\mbox{$t+{}$\textonehalf{}}. The set of all global states is denoted~$\globalstates$. \looseness=-1

What happens in each round is determined by protocols~$\agprotocol{i}{}$  of agents, protocol~$\envprotocol{}$ of the environment, and chance, the latter implemented as the \emph{adversary} part of the environment.
Agent~$i$'s \emph{protocol} $\agprotocol{i}{} \colon \localstates{i} \to \wp(\wp(\actions[i]))\setminus\{\varnothing\}$ provides a range $\agprotocol{i}{r_i(t)}$ of sets of actions based on $i$'s~current local state $r_i(t)$, with the view of achieving some collective goal.  Recall that the global timestamp~$t$  is \emph{not} part of $r_i(t)$. 
The control of all events---correct, byzantine, and system---lies with the environment~$\epsilon$ via its protocol 
 $\envprotocol{} \colon \bbbt \to \wp(\wp(\gevents))\setminus \{\varnothing\}$, which \emph{can} depend on a timestamp $t\in\bbbt$ but \emph{not} on the current state. 
The environment's protocol is thus kept impartial by denying it an agenda based on the global history so far. Other parts of the environment must, however, have access to the global history, in particular, to ensure causality.
Thus, the environment's protocol provides a range $P_\epsilon(t)$ of sets of events.
Protocols~$\agprotocol{i}{}$~and~$\envprotocol{}$ are non-deterministic and always provide at least one option. The choice among the options (if more than one) is arbitrarily made by the already mentioned \emph{adversary} part of the environment. 
It is also required that all events from~$P_\epsilon(t)$ be mutually compatible at time~$t$. These \emph{$t$-coherency} conditions are: (a)~no more than one system event $go(i)$, $\sleep{i}$, and $\hib{i}$ per agent~$i$ at a time; (b)~a correct event perceived as~$e$ by agent~$i$ is never accompanied by a byzantine event that $i$~would also perceive as~$e$, i.e.,~an agent cannot be mistaken about witnessing  an event that \emph{did} happen; (c)~the~GMI of a byzantine sent message is the same as if a copy of the same message were sent correctly in the same round. Note that the prohibition~(b) does not extend to correct actions.\looseness=-1

\begin{figure*}[t]
    \begin{center}
        \scalebox{0.65}{
\tikzset{nodes0/.style={anchor=west},}
\tikzset{nodes1/.style={text width=1.5cm,anchor=west},}
\tikzset{nodes2a/.style={text width=0.5cm,anchor=west},}
\tikzset{nodes2/.style={text width=0.8cm,anchor=west},}
\tikzset{nodes3/.style={text width=1.2cm,anchor=west},}
\tikzset{nodes4/.style={text width=1.3cm,anchor=west},}

\begin{tikzpicture}
    \node[nodes0] at (0,1) (0) {$\run{}{t}$};

    \node[nodes1] at (2.5,3)  (11) {$\agprotocol{n}{\run{n}{t}}$};
    \node[nodes1] at (2.5,2)  (01) {$\dots$};
    \node[nodes1] at (2.5,1)  (1) {$\agprotocol{1}{\run{1}{t}}$};
    \node[nodes1] at (2.5,-2) (2) {$\envprotocol{t}$};

    \draw[->,fill=white] (0) -- (11) node[midway,fill=white]{$\agprotocol{n}{}$};
    \draw[->] (0) -- (1) node[midway,fill=white]{$\agprotocol{1}{}$};
    \draw[->] (0) -- (2) node[midway,fill=white]{$\envprotocol{}$};

    \node[nodes2a] at (7,3)  (13a) {$X_n$};
    \node[nodes2a] at (7,2)  (03a) {$\dots$};
    \node[nodes2a] at (7,1)  (3a) {$X_1$};
    \node[nodes2a] at (7,-2) (4a) {$X_\epsilon$};
    
    \draw[->] (11) -- (13a) node[midway,fill=white] {$\adversary{}$};
    \draw[->] (1) -- (3a) node[midway,fill=white] {$\adversary{}$};
    \draw[->] (2) -- (4a) node[midway,fill=white] {$\adversary{}$};
    
    \node[nodes2] at (9.5,3)  (13) {$\alphaag{n}{r}{t}$};
    \node[nodes2] at (9.5,2)  (03) {$\dots$};
    \node[nodes2] at (9.5,1)  (3) {$\alphaag{1}{r}{t}$};
    \node[nodes2] at (12.5,3)  (13m) {$\alphaag{n}{r}{t}$};
    \node[nodes2] at (12.5,2)  (03m) {$\dots$};
    \node[nodes2] at (12.5,1)  (3m) {$\alphaag{1}{r}{t}$};

    \node[nodes2] at (9.5,-2) (4) {$X_\epsilon= \alphae{r}{t}$};
    
    \draw[->] (13a) -- (13) node[midway,fill=white] {$\mathit{global}$};
    \draw[->] (3a) -- (3) node[midway,fill=white] {$\mathit{global}$};
    \draw[double] (4a) -- (4) node[midway] {};
    \draw[double] (4) -- (4) node[midway] {};
    \node[nodes3] at (16,3) (15) {$\betaag{n}{r}{t}$};
    \node[nodes3] at (16,2) (05) {$\dots$};
    \node[nodes3] at (16,1) (5) {$\betaag{1}{r}{t}$};
    \node[nodes3] at (12.5,-2) (6m) {$\betae{r}{t}$};
    \node[nodes3] at (16,-2) (6) {$\betae{r}{t}$};
    
    \draw[double] (13) -- (13m);
    \draw[double] (3) -- (3m);
    \draw[double] (6) -- (6m);

    \draw[->] (13m) -- (15) node[midway,fill=white] (f3)  {$\filterag{n}{}{}$};
    \draw[->] (3m) -- (5) node[midway,fill=white] (f1) {$\filterag{1}{}{}$};
    \draw[->] (4) -- (6m) node[midway,fill=white] (fe) {$\filtere{}{}$};

    \draw[->] (13) edge (fe) (3) edge (fe);
    \draw[->,dotted] (13) edge (f1) (3) edge (f3);
    \draw[->,dashed] (6m) edge (f1) (6m) edge (f3);    
    
    \node[nodes4] at (19,3) (17) {$\run{n}{t+1}$};
    \node[nodes4] at (19,2) (07) {$\dots$};
    \node[nodes4] at (19,1) (7) {$\run{1}{t+1}$};
    \node[nodes4] at (19,-2) (8) {$\run{\epsilon}{t+1}$};

    \draw[->] (15) -- (17) node[midway,fill=white](19) {$\updateag{n}{}{}{}$};
    \draw[->,fill=white] (5) -- (7) node[midway,fill=white] (9){$\updateag{1}{}{}{}$};
    \draw[->] (6) -- (8) node[midway,fill=white] (a){$\updatee{}{}$};
    
    \draw[->] (5) -- (a);
    \draw[->] (15) -- (a);
    
    \draw[->,dashed] (6) -- (19) node[near end,fill=white] {$\betae[n]{r}{t}$};    
    \draw[->,dashed] (6) -- (9) node[near end,right,fill=white] {$\betae[1]{r}{t}$};

    \node[nodes0] at (21.5,1) (10){$\run{}{t+1}$};
    \draw[->] (7) edge (10) (17) edge (10) (8) edge (10);
    
    \draw[->] (-0.25,-3) -- (22.5,-3);
    \node at (0,-3)    (t0) {$|$};
    \node at (0,-3.5) {$t$};
    \node at (3.1,-3)    (t1) {$|$};
    \node at (7.3,-3)  (t2a) {$|$};
    \node at (9.8,-3)  (t2) {$|$};
    \node at (16.6,-3) (t3) {$|$};
    \node at (22,-3)   (t5) {$|$};
    \node at (22,-3.5) {$t+1$};

    \draw[-] (t0) -- (t1) node[midway,above] {Protocol phase};
    \draw[-] (t1) -- (t2a) node[midway,above] {Adversary phase};
    \draw[-] (t2a) -- (t2) node[midway,above] {Labeling phase};
    \draw[-] (t2) -- (t3) node[midway,above] {Filtering phase};
    \draw[-] (t3) -- (t5) node[midway,above] {Updating phase};
\end{tikzpicture}
        }
        \caption{Details of round~$t+{}$\textonehalf{} of a $\transitionExt{f}{}$-transitional run~$r$.}
        \label{fig:trans_rel}
    \end{center}
\end{figure*}

Both the global run $r \colon \bbbt \to \globalstates$ and its local parts $r_i \colon \bbbt \to \localstates{i}$ provide a sequence of snapshots of the system and local states respectively. Given the \emph{joint protocol} $P\ce(P_1,\dots,P_n)$ and the environment's protocol~$P_\epsilon$, we focus on  \emph{$\transitionExt{f}{}$-transitional runs}~$r$ that result from following these protocols and are built according to a \emph{transition relation} $\transitionExt{f}{}\subseteq \globalstates \times \globalstates$ for asynchronous agents at most  $f\geq0$ of which may become faulty in a given run. In this paper, we only deal with generic~$f$, $P_\epsilon$,~and~$P$. Hence, whenever safe, we write~$\tau$ in place of~$\transitionExt{f}{}$. Each  transitional run begins in some initial global state $r(0)\in\globalinitialstates$ and progresses by ensuring that  $\run{}{t}\,\tau\,\,\run{}{t+1} $, i.e.,~$(\run{}{t},\run{}{t+1}) \in \tau$, for each timestamp  $t\in\bbbt$. 
Given~$f$, $P_\epsilon$,~and~$P$, the transition relation~$\transitionExt{f}{}$ consisting of  five consecutive  phases is graphically represented in Figure~\ref{fig:trans_rel} and described in detail below: \looseness=-1
\begin{description}
\item[\textbf{1. Protocol phase.}]
 A  range $\envprotocol{t}\subseteq\wp(\gevents)$ of  $t$-coherent sets of events is determined by the environment's protocol~$P_\epsilon$;
for each $i\in\agents$, a  range~$\agprotocol{i}{\run{i}{t}}\subseteq\wp(\actions[i])$ of sets of $i$'s~actions is determined by the agents' joint protocol~$P$.\looseness=-1
\item[\textbf{2. Adversary phase.}]       
 The adversary non-deterministically picks  a $t$-coherent set 
        $X_\epsilon \in\envprotocol{t}$ and 
        a set 
        $X_i\in\agprotocol{i}{\run{i}{t}}$ for each $i\in\agents$. \looseness=-1
\item[\textbf{3. Labeling phase.}]
Locally represented actions in~$X_i$'s are translated into the global format: $\alphaag{i}{r}{t}\ce \{\glob{i,t}{a} \mid a \in X_i\}\subseteq \gtrueactions[i]$. In particular, correct sends are supplied with~GMIs.\looseness=-1
\item[\textbf{4. Filtering phase.}]
        Functions $\filtere{}{}{}$ and  $\filterag{i}{}{}$  for each $i \in \agents$ remove all  causally impossible  attempted events from $\alphae{r}{t}\ce X_\epsilon$ and actions from $\alphaag{i}{r}{t}$. \looseness=-1\\
4.1. First, $\filtere{}{}$ filters out causally impossible  events based (a)~on the current global state $r(t)$, which could not have been accounted for by the protocol~$P_\epsilon$, (b)~on $\alphae{r}{t}$, and (c)~on all  $\alphaag{i}{r}{t}$, not accessible for~$P_\epsilon$ either. Specifically, two  kinds of events are causally impossible for asynchronous agents with at most $f$~byzantine failures and are  removed by  $\filtere{}{}$ in two stages as follows (formal definitions can be found in the appendix, Definitions~\ref{def:filter_env}--\ref{def:filter_ag}; cf.~also~\cite{PKS19:TR} for details): \looseness=-1
\begin{compactenum}[(1)]
\item in the 1st~stage, all byzantine events are removed by $\filtere[\leq f]{}{}$ if they would have resulted in more than $f$~faulty agents in total;\looseness=-1
\item in the 2nd~stage, correct receives  without matching sends (either in the history $r(t)$ or in the current round) are removed by $\filtere[B]{}{}$.\looseness=-1
\end{compactenum}    
The  resulting set of events to actually occur in round~\mbox{$t+{}$\textonehalf{}} is denoted \looseness=-1
        \[
        \label{eq:beta_env_run}
        \betae{r}{t} \ce \filtere{\run{}{t}}{\quad\alphae{r}{t},\quad \alphaag{1}{r}{t}, \quad\dots ,\quad \alphaag{n}{r}{t}}.
        \]
4.2. After events are filtered, $\filterag{i}{}{}$ for each agent $i$  removes all $i$'s actions  if{f} $go(i) \notin \betae{r}{t}$.  The resulting sets of actions to be actually performed by agents in round $t+{}$\textonehalf{} are\looseness=-1 
        \[
        \label{eq:beta_agent_run}
        \betaag{i}{r}{t} \ce \filterag{i}{\alphaag{1}{r}{t}, \quad\dots ,\quad \alphaag{n}{r}{t} }{\quad\betae{r}{t}}.
        \]
                We have $\betaag{i}{r}{t} \subseteq \alphaag{i}{r}{t} \subseteq \gtrueactions[i]$ and $\betae{r}{t} \subseteq \alphae{r}{t} \subseteq \gevents$.\looseness=-1
\item[\textbf{5. Updating phase.}]
         The resulting mutually causally consistent sets of events  $\betae{r}{t}$ and of actions~$\betaag{i}{r}{t}$ are appended to the global history $r(t)$; for each $i \in \agents$, all non-system events from 
         \[
         \betae[i]{r}{t} \ce \betae{r}{t}\cap\gevents[i]
         \]
          as \emph{perceived} by the agent and all correct actions $\betaag{i}{r}{t}$ are appended \emph{in the local form} to the local history $r_i(t)$, which may remain unchanged if no action or event triggers an update or be appended with the empty set if an update is triggered only by a system event $go(i)$ or $\sleep{i}$:\looseness=-1
            \begin{align}
            \label{eq:run_trans_env}
        \run{\epsilon}{t+1} \ce& \updatee{\run{\epsilon}{t}}{\quad\betae{r}{t},\quad \betaag{1}{r}{t},\quad \dots,\quad \betaag{n}{r}{t}};
\\
	\label{eq:run_trans_ag}
        \run{i}{t+1} \ce& \updateag{i}{\run{i}{t}}{\quad\betaag{i}{r}{t}}{\quad\betae{r}{t}}.  
         \end{align}
         Formal definitions of $\updatee{}{}{}$ and $\updateag{i}{}{}{}$ are given  in Def.~\ref{def:state-update} in the appendix.
\end{description}

The  protocols~$P$~and~$P_\epsilon$ only affect phase~1, so we group the operations in the remaining phases~2--5 into a \emph{transition template}~$\tau_f$ that computes a transition relation~$ \transitionExt{f}{}$ for any given~$P$~and~$P_\epsilon$. This transition template, primarily via the filtering functions, represents asynchronous agents with at most $f$~faults. The template can be modified independently from the protocols to capture other distributed scenarios.\looseness=-1

 \emph{Liveness} and similar  properties that cannot be ensured on a round-by-round basis are enforced by restricting the allowable runs by \emph{admissibility conditions}~$\Psi$, which formally are  subsets of the set~$R$ of all transitional runs. 
For example,
since no goal can be achieved without allowing agents to act from time to time, it is standard to impose the
\emph{Fair Schedule}~($\FS$) admissibility condition, which for byzantine agents  states that an agent can only be  delayed indefinitely through persistent faults:
\[\FS \ce \{ r \in R \mid (\forall i \in \agents)\, (\forall t\in\bbbt)\,(\exists t'\ge t)\, 
             {\betae{r}{t'} \cap \sevents[i]}\ne \varnothing
        \}.
        \]
In scheduling terms,  $\FS$~ensures that each agent be considered for using  CPU~time infinitely often. Denying any of these requests constitutes a failure, represented by a $\sleep{i}$ or $\hib{i}$ system event. \looseness=-1

We now combine all these parts in the notions of context and agent-context:

\begin{definition}
    \label{def:consistency}  \label{def:system}  
A \emph{context}  $\gamma=(\envprotocol{},\globalinitialstates,\tau_f,\Psi)$ consists of   an environment's protocol~$\envprotocol{}$, a set of global initial states $\globalinitialstates$,  a transition template~$\tau_f$ for $f\geq 0$, and an admissibility condition~$\Psi$. For a joint protocol~$\joinprotocol{}$, we call $\chi= (\gamma,\joinprotocol{})$  an \emph{agent-context}. A run $r \colon \bbbt \to \globalstates$ is called \emph{weakly $\chi$-consistent} if $r(0) \in \globalinitialstates$ and the run is $ \transitionExt{f}{}$-transitional. A weakly $\chi$-consistent run~$r$ is called \emph{(strongly)} $\chi$-\emph{consistent} if  $r \in \Psi$. The set of all $\chi$-consistent runs  is denoted~$\system{\chi}\subseteq R$. Agent-context~$\chi$ is called \emph{non-excluding} if any finite prefix of a  weakly $\chi$-consistent run  can be extended to a  $\chi$-consistent run.\looseness=-1
   \end{definition}
We  are also interested in  narrower types of faults. Let $\fevents[i]\ce \bevents[i]\sqcup\{\sleep{i},\hib{i}\}$. 
 \looseness=-1
\begin{definition} 
\label{def:agent_types}Environment's protocol~$\envprotocol{}$ makes an agent $i \in \agents$:
\begin{compactenum}
\item 
\emph{correctable} if  $X \in \envprotocol{t}$ implies that  
$
   X \setminus \fevents[i]    \in \envprotocol{t}$;\looseness=-1

\item \emph{delayable} if     $X \in \envprotocol{t}$ implies
   $X \setminus \gevents[i]  \in \envprotocol{t}$;\looseness=-1

\item \emph{error-prone} if
$X \in \envprotocol{t}$ implies that, for any $Y \subseteq  \fevents[i]$, the set
   $Y \sqcup (X \setminus  \fevents[i])  \in \envprotocol{t}$ whenever it is $t$-coherent;\looseness=-1
   
   \item\emph{gullible} 
 if
    $X \in \envprotocol{t} $
    implies that, for any $Y \subseteq  \fevents[i]$,  the set
   $Y \sqcup (X \setminus \gevents[i])  \in \envprotocol{t}$ whenever it is $t$-coherent;\looseness=-1
\item \emph{fully byzantine} if agent $i$ is both error-prone and gullible.\looseness=-1
\end{compactenum}
\end{definition}
In other words, correctable agents can always be made correct for the round by removing all their byzantine events; delayable agents can always be forced to skip a round completely (which does not make them faulty); error-prone (gullible) agents can exhibit any faults in addition to (in place of) correct events, thus, implying correctability (delayability);  fully byzantine agents' faults are unrestricted. Common types of faults, e.g.,~crash or omission failures, can be obtained by restricting allowable sets~$Y$ in the definition of gullible agents.  \looseness=-1

Now that our byzantine version of the runs-and-systems framework is laid out, we define interpreted systems in this framework in the usual way, i.e., as  special kinds of Kripke models for multi-agent distributed environments~\cite{bookof4}. 
For an agent-context~$\chi$, we consider pairs $(r,t')\in R^\chi\times \bbbt$ of a $\chi$-consistent run~$r$  and timestamp~$t'$. A \emph{valuation function} $\pi \colon \prop \to \wp(R^\chi \times \bbbt)$ determines  whether an atomic proposition from $\prop$ is true in run~$r$ at time~$t'$. The determination is arbitrary except for a small set of \emph{designated atomic propositions} 
whose truth value at~$(r,t')$ is fully determined. More specifically, for $i\in\agents$, $o \in \haps[i]$, and $t\in \bbbt$ such that $t \leq t'$, \looseness=-1

$\correct{(i,t)}$ is true at~$(r,t')$, or \emph{node~$(i,t)$ is correct in  run~$r$}, if{f} no faulty event happened to~$i$  by timestamp~$t$, i.e.,~no event from $\fevents[i]$ appears in the $r_\epsilon(t)$ prefix of the $r_\epsilon(t')$ part of $r(t')$; \looseness=-1

$\correct{i}$ is true at~$(r,t')$ if{f} $\correct{(i,t')}$ is;\looseness=-1

  $\fake[(i,t)]{o}$ is true at~$(r,t')$ if{f} $i$~has a \emph{faulty} reason to believe  that  $o\in\haps[i]$ occurred in   round~\mbox{$t-{}$\textonehalf}, i.e.,~$o \in r_i(t)$ because (at least in part) of some $O \in \bevents[i] \cap \betae{r}{t-1}$;\looseness=-1
  
$\trueoccurred[(i,t)]{o}$ is true at~$(r,t')$ if{f} $i$ has a \emph{correct} reason to believe    $o\in\haps[i]$ occurred in   round~\mbox{$t-{}$\textonehalf}, i.e.,~$o \in r_i(t)$ because (at least in part) of $O \in (\gtrueevents[i] \cap \betae{r}{t-1}) \sqcup  \betaag{i}{r}{t-1}$;\looseness=-1

   $\trueoccurred[i]{o}$ is true at~$(r,t')$ if{f} at least one of $\trueoccurred[(i,m)]{o}$ for $1 \leq m \leq t'$ is; 
   
   $\trueoccurred{o}$  is true at~$(r,t')$ if{f} at least one of $\trueoccurred[i]{o}$ for $i \in \agents$ is;\looseness=-1
   
   $\occurred[i]{o}$  is true at~$(r,t')$ if{f} either $\trueoccurred[i]{o}$ is or at least one of $\fake[(i,m)]{o}$  for $1 \leq m \leq t'$ is. \looseness=-1 

An \emph{interpreted system} is a pair $\I = (R^\chi, \pi)$. The epistemic language $\varphi \cce p \mid \lnot \varphi \mid (\varphi \land \varphi) \mid K_i \varphi$ where $p \in \prop$ and $i\in\agents$ and derived Boolean connectives are defined in the usual way.  Truth  for these (\emph{epistemic}) \emph{formulas} is defined in the standard way, in particular, for a run $r \in R^\chi$, timestamp $t \in \bbbt$, atomic proposition $p \in \prop$, agent $i \in \agents$, and formula $\varphi$ we have 
$(\I,r,t) \models p$ 
if{f} 
$(r,t) \in \pi(p)$ and 
$(\I,r,t) \models K_i \varphi$ if{f} $(\I,r',t') \models \varphi$ 
for any $r'\in R^\chi$ and $t' \in \bbbt$ such that 
$r_i(t) = r'_i(t')$. A formula~$\varphi$ is valid in~$\I$, written $\I \models \varphi$, if{f} 
$(\I,r,t) \models \varphi$ for all $r \in R^\chi$ and $t \in \bbbt$.\looseness=-1

Due to the $t$-coherency of all allowed protocols~$P_\epsilon$, an agent cannot be both right and wrong about any local  event $e \in \events[i] $, i.e.,~$\I \models 
\lnot(\trueoccurred[(i,t)]{e} \land \fake[(i,t)]{e})$. Note that for actions this \emph{can} happen.\looseness=-1

Following the concept   from~\cite{CKRev} of  global events that are local for an agent, we define conditions under which formulas can be treated as such local events. A formula~$\varphi$ is called  \emph{localized for~$i$ within an agent-context~$\chi$} if{f} 
   $r_i(t) = r'_i(t')$ implies
   $(\I,r,t) \models \varphi \Longleftrightarrow (\I,r',t') \models \varphi$
   for any $
  \I=(R^\chi,\pi)$, runs $r, r' \in R^\chi$, and timestamps $t, t' \in \bbbt$.
By these definitions, we immediately obtain:\looseness=-1
\begin{lemma}
\label{lem:localized}
The following statements are valid for any formula~$\varphi$ localized for an agent $i\in\agents$ within an agent-context~$\chi$ and  any interpreted system $\I=(R^\chi,\pi)$: $
 \I \models \varphi \leftrightarrow K_i\varphi$ and
$ \I \models \lnot\varphi \leftrightarrow K_i\lnot\varphi$.
\end{lemma}
The knowledge of preconditions principle~\cite{Mos15TARK} postulates that in order to  act on a  precondition an agent must be able to infer it from its local state. Thus, Lemma~\ref{lem:localized}
 shows that formulas localized for~$i$ can \emph{always} be used as preconditions.
Our first observation is that the agent's \emph{perceptions} of a run are one example of such epistemically acceptable  (though not necessarily reliable) preconditions:\looseness=-1
\begin{lemma}
    \label{lem:occurred-persistence_paper}
    For any agent-context~$\chi$, agent $i \in \agents$, and local hap $o \in \haps[i]$, the formula $\occurred[i]{o}$ is localized for~$i$ within~$\chi$.\looseness=-1
\end{lemma}
It can be shown that correctness of these perceptions is not localized for $i$ and, hence, cannot be the basis for actions. In fact, Theorem~\ref{lem:no-k-occurred} will reveal that no agent can 
establish its own correctness.

\section{Run modifications}
\label{sect:run_mod}

We will now introduce the pivotal technique of \emph{run modifications}, which are used to show an agent does not know~$\varphi$ by creating an indistinguishable run where $\varphi$~is false.
\looseness=-1

\begin{definition}
A function $\rho \colon  R^\chi \rightarrow \wp(\gtrueactions[i]) \times \wp(\gevents[i])$ is called an \emph{$i$-intervention for an agent-context~$\chi$ and agent $i\in\agents$}. 
A \emph{joint intervention} $B= (\rho_1,\dots,\rho_n)$ consists of $i$-interventions~$\rho_i$ for each agent $i \in \agents$. 
An \emph{adjustment} $[B_t;\dots;B_0]$ is a sequence of joint interventions
$B_0\,\dots,B_t$ 
to be performed at rounds from~\textonehalf{} to~\mbox{$t+{}$\textonehalf{}}. \looseness=-1
\end{definition}
We consider an  $i$-intervention $\rho(r)=(X, X_{\epsilon})$ applied to a round~\mbox{$t+{}$\textonehalf{}} of a given run~$r$ to be  a meta-action by the system designer, intended to modify the results of this round for~$i$ in such a way that  
$\betaag{i}{r'}{t}=X$ 
and 
$\betae[i]{r'}{t}=\betae{r'}{t} \cap \gevents[i]=X_{\epsilon}$ 
in the artificially constructed new run~$r'$. For $\rho(r)=(X, X_{\epsilon})$,
we denote $\A\rho(r) \ce  X$ and $\E\rho(r) \ce X_{\epsilon}$. Accordingly, a joint intervention $(\rho_1,\dots,\rho_n)$ prescribes actions $\betaag{i}{r'}{t}=\A\rho_i(r)$  for each agent~$i$ and events $\betae{r'}{t}=\bigsqcup_{i\in\agents} \E\rho_i(r)$ for the round in question. Thus, an adjustment $[B_t;\dots;B_0]$ fully determines actions and events in the initial \mbox{$t+1$}~rounds of run~$r'$:\looseness=-1
\begin{definition}
    \label{def:build_from}
    Let
    $\label{eq:generic_adj}
    \adj = \left[
    B_t;\dots;B_0
    \right]$  be an adjustment  where
    $
    B_{m} = (\rho_1^{m}, \dots, \rho_n^{m})$
    for each $0 \leq m \leq t$ 
    and each~$\rho_i^{m}$ is an $i$-intervention for an agent-context $\chi=\left((\envprotocol{},\globalinitialstates,\tau_f,\Psi),P\right)$.
A run~$r'$ is \emph{obtained from}  $r \in R^\chi$ \emph{by adjustment} $\adj$ if{f} for all $t'\leq t$, all $T' > t$, and all $i \in \agents$, \looseness=-1
\begin{enumerate}[(a)]
\item           
	$\run[']{}{0}\ce\run{}{0}$,  
\item\label{clause:run_adj_ag}               
	$
            \run[']{i}{t'+1} \ce \mathit{update}_i\left(\run[']{i}{t'},\,\,\A\rho_i^{t'}(r),\,\,\bigsqcup_{i \in\agents}\E\rho_i^{t'}(r)\right)
        $,
\item                
	$
            \run[']{\epsilon}{t'+1} \ce \mathit{update}_\epsilon\left(\run[']{\epsilon}{t'},\,\,\bigsqcup_{i \in\agents}\E\rho_i^{t'}(r),\,\, \A\rho_1^{t'}(r),\,\, \dots,\,\,\A\rho_n^{t'}(r)\right)$,
\item                     
	$
            r'(T')\,\,\transitionExt{f}{}\,\,r'(T'+1)$.
                        \end{enumerate}
$\SystemBuildFrom{\transitionExt{f}{}}{r}{\adj}$ is the set of all runs obtained from~$r$ by $\adj$.\looseness=-1
\end{definition}

Note that adjusted  runs need not be a priori transitional, i.e.,~obey~(d), for $t' \leq t$. Of course, we intend to use adjustments in such a way that $r'$~\emph{is} a transitional run. But it requires a separate proof.  In order to improve the readability of these proofs, we allow ourselves (and already used) a small abuse of notation. The $\beta$-sets $\betae{r'}{t}$  and $\betaag{i}{r'}{t}$ were initially defined only for transitional runs as the result of filtering. But they also represent the sets of events
and $i$'s~actions  respectively happening in round~\mbox{$t+{}$\textonehalf}. This alternative definition is equivalent for transitional runs and, in addition, can be used for adjusted runs~$r'$. This is what we mean whenever we write $\beta$-sets for runs obtained by adjustments.

In order to minimize an agent's knowledge in accordance with the structure of its (soon to be defined) reliable (or byzantine) causal cone, we will use several types of $i$-interventions that copy round~\mbox{$t+{}$\textonehalf{}} of the original run to various degrees: (a)~$\newfreezerule{}$  denies~$i$ all actions and events, (b)~$\newechorule{i}{t}{}{}$ reproduces  all messages sent by~$i$ but  in byzantine form, (c)~$\newchatterrule{i}{t}{X}{}$ (for an appropriately chosen set $X \subseteq \agents \times \bbbt$) faithfully reproduces  all actions and as many events  as causally possible.\looseness=-1

\begin{definition}
    \label{def:ifakerule}
    For an agent-context~$\chi$, $i\in\agents$, and  $r\in R^\chi$, we define the following $i$-interventions: \begin{equation}
    \newfreezerule{r} \ce (\varnothing, \varnothing).
    \end{equation}  
   \vskip-5ex
    \begin{multline}
        \label{eq:echorule}
        \newechorule{i}{t}{r} \ce 
        \Bigl(\varnothing ,\qquad
        \{\failof{i}\} \quad \sqcup  \quad
        \bigl\{\,\fakeof{i}{\mistakefor{\gsend{i}{j}{\mu}{id}}{\tick}}  \qquad\big|\bigr. \\
        \bigl.\gsend{i}{j}{\mu}{id}\in \betaag{i}{r}{t}\quad\lor\quad (\exists A)\,         \fakeof{i}{\mistakefor{\gsend{i}{j}{\mu}{id}}{A}} \in \betae{r}{t} \bigr\}\Bigr) .
    \end{multline}
    \vskip-5ex
        \begin{multline}
        \label{eq:chatterrule}
        \newchatterrule{i}{t}{X}{r} \ce \Bigl(\betaag{i}{r}{t} ,\qquad \betae[i]{r}{t}\setminus \bigl\{\grecv{i}{j}{\mu}{id(j,i,\mu,k,m)} \quad\big|\quad (j,m) \notin X,\, k \in \mathbb{N}\bigr\}\Bigr).
    \end{multline}
\end{definition}

\section{The Reliable Causal Cone}
\label{sect:ByzCausalCone}

Before giving formal definitions and proofs, we first explain the intuition behind our byzantine analog of Lamport's causal cone, and the particular adjustments used for constructing runs with identical reliable causal cones in the main Lemma~\ref{lem:sources_isolation} of this section.
\looseness=-1

In the absence of faults~\cite{ChaMis86DC}, the only way information, say, about an event~$e$ that happens to an agent~$j$ can reach another agent~$i$, or more precisely, its node~$(i,t)$, is via a causal chain (time progression and delivered messages) originating from~$j$ after $e$~happened and reaching~$i$ no later than timestamp~$t$. The set of beginnings of  causal chains, together with all causal links, is called the \emph{causal cone} of~$(i,t)$. The standard way of demonstrating the necessity of such a chain for~$i$ to learn about~$e$, when expressed in our terminology, is  by using an adjustment that removes all events and actions outside the causal cone. Once  an adjusted run with no haps outside the causal cone is shown to be transitional
and the local state of~$i$ at timestamp~$t$ is shown to be the same as in the given run,  it follows that $i$~considers it possible that $e$~did not happen and, hence, does not know  that $e$~happened. This well known proof is carried out   in our framework  in~\cite{FroCoS19}  (see also~\cite{PKS19:TR} for an extended version). \looseness=-1

However,  one subtle aspect of our formalization is also relevant for the byzantine case. 
We illustrate it using a minimal example. Suppose, in the given run, $j_s$~sent exactly one message to~$j_r$ during round~\mbox{$m+{}$\textonehalf{}}  and it was correctly received by~$j_r$ in round~\mbox{$l+{}$\textonehalf{}}. At the same round, $j_r$~itself sent its last ever message, and sent it to~$i$. If this message to~$i$ arrived before~$t$, then  $(j_r,l)$~is a  node within the causal cone of~$(i,t)$. On the other hand, neither~$(j_r,l+1)$~nor~$(j_s,m)$ are within the causal cone. Thus, the run adjustment discussed in the previous paragraph removes the action of sending the message from~$j_s$ to~$j_r$, which happened outside the causal cone, and, hence, makes it causally impossible for~$j_r$ to receive it  despite the fact that the receipt happened within, or more precisely, on the boundary of the causal cone. On the other hand, the message sent by~$j_r$ in the same round cannot be suppressed without $i$~noticing. Thus, suppressing all haps on the boundary of the causal cone is not an option. These considerations necessitate the use of $\newchatterrule{j}{l}{X}{}$ to remove such ``ghosts'' of messages instead of the exact copy of round~\mbox{$l+{}$\textonehalf{}} of the given run. To obtain Chandy--Misra's result, one needs to set $X$~to be the entire causal cone.\footnote{This treatment of the cone's boundary could be perceived as overly pedantic. But in our view this is preferable to being insufficiently precise.}\looseness=-1

We now explain the complications created by the presence of byzantine faults. Because byzantine agents can lie, the existence of a causal chain is no more sufficient for reliable delivery of information. 
Causal chains can now be  \emph{reliable}, i.e.,~involve only correct agents, or \emph{unreliable}, whereby a byzantine agent can corrupt the transmitted information or even initiate the whole communication while pretending to be part of a longer chain. 
If several causal chains  link a node~$(j,m)$ witnessing an event with~$(i,t)$, where the decision based on this event is to be made,  then, intuitively, the information about the event can only be relied upon if at least one of these causal chains is reliable.
In effect, all correct nodes, i.e.,~nodes $(j,m)$ such that $(\I,r,t) \vDash \correct{(j,m)}$, are divided into three categories: those without any causal chains to~$(i,t)$, i.e.,~nodes outside Lamport's causal cone, those with causal chains but only  unreliable ones, and those with at least one reliable causal chain. 
There is, of course, the fourth category consisting of byzantine nodes, i.e.,~nodes~$(j,m)$ such that  $(\I,r,t) \nvDash \correct{(j,m)}$. 
Since there is no way for nodes without reliable causal chains  to make themselves heard, we call these nodes \emph{silent masses} and apply  to them the $\newfreezerule{}$ intervention: since they cannot have an impact, they need not act.
The nodes with at least one reliable causal chain to~$(i,t)$, which must be correct themselves, form the \emph{reliable causal cone} and  treated the same way as Lamport's causal cone in the fault-free case, except that the removal of ``ghost'' messages is more involved in this case. Finally, the remaining nodes are byzantine and form a \emph{fault buffer} on the way of reliable information. Their role is to pretend the run is the same independently of what the silent masses do. 
We will show that  $\newechorule{j}{m}{}{}$ suffices since only messages sent from the fault buffer matter.
  \looseness=-1

Before stating our main Lemma~\ref{lem:sources_isolation}, which constructs an adjusted run that leaves agent~$i$ at~$t$ in the same position while removing as many haps as possible, it should be noted that our analysis relies on \emph{knowing which agents are byzantine} in the given run, which may easily change without affecting local histories. This assumption will be dropped in the following section.
\looseness=-1

First we define simple causal links among nodes as binary relations on  $\agents\times \bbbt$ in  infix notation:\looseness=-1
\begin{definition}
    \label{def:causal_graph}
    For all $ i \in \agents$ and $t \in \bbbt$, we have
    $(i,t)\edgLoc(i,t+1)$.
    Additionally, for a run~$r$, we have 
    $(i,m)\edgMsg{r}(j,l)$ if{f} 
    there are $\mu \in Msgs$ and $id \in \bbbn$ such that 
    $\grecv{j}{i}{\mu}{id}\in \betae{r}{l-1}$ and 
    either $\gsend{i}{j}{\mu}{id}\in \betaag{i}{r}{m}$ or 
    $\fakeof{i}{\mistakefor{\gsend{i}{j}{\mu}{id}}{A}} \in \betae{r}{m}$ 
    for some $A \in \{\tick\}\sqcup\gactions[i]$. 
    \emph{Causal $r$-links} $\edgC{r} \ce \edgLoc \cup \edgMsg{r}$ are either local or communication related. 
    A \emph{causal $r$-path}  for a run~$r$ is a sequence $\xi = \langle\theta_0,\theta_1,\dots,\theta_k\rangle$, $k\geq 0$, of nodes connected by causal $r$-links, i.e., such that $\theta_l \edgC{r} \theta_{l+1}$ for each $0 \leq l <k$. 
    This causal $r$-path is called \emph{reliable} if{f} node $(j_l,t_l+1)$ is correct in~$r$ for each $\theta_l = (j_l,t_l)$ with $0\leq l < k$ and, additionally, node $\theta_k= (j_k,t_k)$ is correct in~$r$. 
    We also write $\pastpath[\xi]{r}{\theta_0}{\theta_k}$ to denote the fact that path~$\xi$ connects node~$\theta_0$ to~$\theta_k$ in run~$r$, or simply $\pastpath{r}{\theta_0}{\theta_k}$ to state that such a causal $r$-path exists.\looseness=-1
\end{definition}
Note that neither receives nor sends of messages forming a reliable causal $r$-path can be byzantine. The latter is guaranteed by the immediate future of nodes on the path being correct.
    \begin{definition}    
    The \emph{reliable causal cone~$\truecausalcone{r}{\theta}$  of node~$\theta$ in run~$r$}          consists of all nodes $\zeta \in \agents \times\bbbt$ such that $\pastpath[\xi]{r}{\zeta}{\theta}$ for some reliable causal $r$-path~$\xi$. The \emph{fault buffer~$\faultbuffer{r}{\theta}$ of node~$\theta$ in run~$r$} consists of all nodes~$(j,m)$ with $m<t$ such that $\pastpath{r}{(j,m)}{\theta}$ and  $(j,m+1)$ is not correct. Abbreviating $\havevoice{r}{\theta}\ce\truecausalcone{r}{\theta} \sqcup \faultbuffer{r}{\theta}$, the  \emph{silent masses of node~$\theta$ in run~$r$} are all the remaining nodes $\silentmasses{r}{\theta}\ce(\agents \times \bbbt) \setminus \havevoice{r}{\theta}$.\looseness=-1
\end{definition}

Here the filling of the cone~$\truecausalcone{}{}$ signifies reliable communication, $\faultbuffer{}{}$~represents a barrier for correct information, whereas $\silentmasses{}{}$~depicts correct information isolated from its destination.
We can now state the main result of this section:

\begin{lemma}[Cone-equivalent run construction]
    \label{lem:sources_isolation}
    For $f \in \mathbb{N}$, for a non-excluding agent-context  $\chi=\bigl((\envprotocol{}, \globalinitialstates, \tau_f,\Psi),\joinprotocol{}\bigr)$ such that all agents     are gullible, correctable, and delayable,   
    for any $\transitionExt{f}{}$-tran\-sit\-ional run~$r$, 
    for a node $\theta=(i,t)\in\points$ correct in~$r$,
    let adjustment $\adj = [B_{t-1}; \dots ;B_0]$ where
    $
    B_{m} = (\rho_1^{m}, \dots, \rho_n^{m})$
    for each $0 \leq m \leq  t-1$ such that\looseness=-1
    \begin{equation}
    \label{eq:Byz_causal_adj}
    \rho^m_j \ce 
    \begin{cases}
   \newchatterrule{j}{m}{ \havevoice{r}{\theta}}{}& \text{if $(j,m)\in \truecausalcone{r}{\theta}{}  $},
    \\
   \newechorule{j}{m}{} & \text{if $(j,m)\in \faultbuffer{r}{\theta}$},
    \\
    \newfreezerule{}{}{} & \text{if $(j,m)\in  \silentmasses{r}{\theta} $}.
    \end{cases}
    \end{equation}
  Then
    each  $r'\in\SystemBuildFrom{\transitionExt{f}{}}{r}{\adj}$ satisfies the following properties: \looseness=-1
    \begin{compactenum}[(A)]
        \item\label{lem:multipede:same} $(\forall (j,m)\in \truecausalcone{r}{\theta})\quad \run[']{j}{m}=\run{j}{m}$;\looseness=-1
\item\label{lem:multipede:same:i} 
$(\forall m\leq t)\quad r'_i(m) = r_i(m)$;
        \item\label{lem:multipede:bad} 
        for any $m\leq t$, we have that  $\betae{r'}{m-1} \cap \fevents[j] \ne \varnothing$  if{f} both $\pastpath{r}{(j,m-1)}{\theta}$ and  $(j,m)$ is not correct in~$r$;\looseness=-1
                \item\label{lem:multipede:nomorefaulty} for any $m \leq t$, any node~$(j,m)$ correct in~$r$ is also correct in~$r'$; \looseness=-1
                \item\label{lem:multipede:nomorethanf}  the number of agents byzantine by any $m \leq t$ in run~$r'$ is not greater than that in run~$r$ and is~$\leq f$;\looseness=-1
      \item\label{lem:multipede:trans} $r'$~is $\transitionExt{f}{}$-transitional.\looseness=-1
               \end{compactenum}
\end{lemma}

\begin{proof}[Proof sketch] The following  properties follow from the definitions: 
\begin{gather}
\truecausalcone{r}{\theta} \cap {\faultbuffer{r}{\theta}} =\varnothing, \qquad\qquad \theta \in \truecausalcone{r}{\theta},
\\
\label{eq:causal_past_of_well_connected}
(j,m)\in\truecausalcone{r}{\theta} \quad \& \quad (k,m') \edgC{r} (j,m) \qquad \Longrightarrow\qquad (k,m') \in {\textstyle\havevoice{r}{\theta}}.
\end{gather} 
 \looseness=-1

Note that for $\beta_k = (j_k,m_k)$ with $k=1,2$, we have $\beta_1 \edgC{r} \beta_2$ implies $m_1 < m_2$ and $\pastpath{r}{\beta_1}{\beta_2}$ implies $m_1 \leq m_2$. Thus, all parts of the lemma except for Statement~\eqref{lem:multipede:trans} only concern $m \leq t$, and even this last statement for $m>t$ is a trivial corollary of Def.~\ref{def:build_from}(d). Thus, we focus on  $m\leq t$.\looseness=-1

Statement~\eqref{lem:multipede:same} can be proved by induction on~$m$ using the following auxiliary lemma for the given transitional run~$r$ and the adjusted run~$r'$, which is also constructed using the standard update functions.\looseness=-1
\begin{lemma}
\label{lem:agree_on}
If $r_j(m) = r'_j(m)$, and~$\betaag{j}{r}{m}=\A\rho_j^{m}(r)$, and~$\betae[j]{r}{m}=\E\rho_j^{m}(r)$, then $r_j(m+1)=r'_j(m+1)$.\looseness=-1
\end{lemma}
\begin{proof}
This statement follows from~\eqref{eq:run_trans_ag} for the transitional run~$r$, Def.~\ref{def:build_from}\eqref{clause:run_adj_ag} for the adjusted run~$r'$, and the fact that  $\updateag{j}{}{}{}$ only depends on events of agent~$j$, in particular, on the presence of $go(j)$ or $\sleep{j}$ (see Def.~\ref{def:state-update} in the appendix for details).\looseness=-1
\end{proof}
The third condition of Lemma~\ref{lem:agree_on} is satisfied for 
$\rho_j^{m}(r)=\newchatterrule{j}{m}{ \havevoice{r}{\theta}}{}$  
within~$\truecausalcone{r}{\theta}$ by~\eqref{eq:causal_past_of_well_connected}. Further, if $(j,m) \in \truecausalcone{r}{\theta}$, then so are all~$(j,m')$ with $m'\leq m$. In particular, $(i,m') \in \truecausalcone{r}{\theta}$ for any $m' \leq t$. Thus, Statement~\eqref{lem:multipede:same:i} follows from Statement~\eqref{lem:multipede:same} we have already proved.\looseness=-1

Statement~\eqref{lem:multipede:bad} is due to the fact that (a)~$\newchatterrule{j}{m}{ \havevoice{r}{\theta}}{}$ 
does not produce any new byzantine events relative to $\betae[j]{r}{m}$, which contains none for $(j,m) \in \truecausalcone{r}{\theta}$, (b)~$\newfreezerule{}{}{}$ never produces byzantine events, whereas (c)~$\newechorule{j}{m}{}$ always contains at least $\failof{j}\in\bevents[j]$. Statements~\eqref{lem:multipede:nomorefaulty} and~\eqref{lem:multipede:nomorethanf} are direct corollaries of Statement~\eqref{lem:multipede:bad}.\looseness=-1

The bulk of the proof concerns Statement~\eqref{lem:multipede:trans}, or, more precisely the transitionality up to timestamp~$t$.  For each  $m<t$, we need sets $\alphae{r'}{m} \in \envprotocol{m}$ and $\alphaag{j}{r'}{m}= \{\glob{j,m}{a} \mid a \in X_j\}$ for some $X_j\in P_j(r'_j(m))$ for each $j \in \agents$ such that for $\betae{r'}{m} = \bigsqcup_{j \in\agents}\E\rho_j^{m}(r)$ and $ \betaag{j}{r'}{m} = \A\rho_j^{m}(r)$ for all $j \in \agents$,\looseness=-1
        \begin{gather}
        \label{eq:trans_filter_desir_env}
         \betae{r'}{m} = \filtere{r'(m)}{\alphae{r'}{m}, \alphaag{1}{r'}{m}, \dots , \alphaag{n}{r'}{m}},
         \\
        \label{eq:trans_filter_desir_ag}
        \betaag{j}{r'}{m} = \filterag{j}{\alphaag{1}{r'}{m}, \dots , \alphaag{n}{r'}{m} }{\betae{r'}{m}}.
        \end{gather}
         The construction of such $\alpha$-sets and the proof of~\eqref{eq:trans_filter_desir_env}--\eqref{eq:trans_filter_desir_ag} for them is by induction on~$m$. Note that $r'_\epsilon(0)=r_\epsilon(0)$ and $r'_j(0)=r_j(0)$ for all $j \in \agents$ by Def.~\ref{def:build_from}(a). We will show that it suffices to choose\looseness=-1
        \begin{equation}
        \label{eq:choosing_new_alpha_ag}
         \alphaag{j}{r'}{m} \ce
         \begin{cases}
         \alphaag{j}{r}{m} & \text{if $(j,m) \in \truecausalcone{r}{\theta}$},
         \\
         \{\glob{j,m}{a} \mid a \in X_j\}
         \text{ for some $X_j \in P_j(r'_j(m))$} & \text{otherwise},
         \end{cases}
         \end{equation}
         with the choice in the latter case possible by $P_j(r'_j(m)) \ne \varnothing$, and 
       \looseness=-1
  \begin{multline}
        \label{eq:choosing_new_alpha_env}
        \alphae{r'}{m} \ce
\Bigl(\filtere[\leq f]{r(m)}{\alphae{r}{m}, \alphaag{1}{r}{m}, \dots, \alphaag{n}{r}{m}}
 \setminus 
\!\!\!\!\bigsqcup\limits_{\scriptscriptstyle(l,m) \in\silentmasses{r}{\theta} \sqcup \faultbuffer{r}{\theta} }\!\!\!\! \gevents[l] \Bigr)
\sqcup 
\{\failof{l}\mid (l,m) \in {\textstyle\faultbuffer{r}{\theta}}\} \sqcup \\
\Bigl\{\fakeof{l}{\mistakefor{\gsend{l}{j}{\mu}{id}}{\tick}}  \quad\Big|\quad
(l,m) \in {\textstyle\faultbuffer{r}{\theta}}\quad\&\quad
\\
\Bigl(\gsend{l}{j}{\mu}{id}\in \betaag{l}{r}{m}  \quad\lor\quad
        (\exists A 
        )\,\fakeof{l}{\mistakefor{\gsend{l}{j}{\mu}{id}}{A}} \in \betae{r}{m}\Bigr)  \Bigr\}.
\end{multline}
Informally, according to~\eqref{eq:choosing_new_alpha_ag}, in~$r'$, we just repeat the choices made in~$r$ within the reliable causal cone and make  arbitrary choices elsewhere. 
 According to~\eqref{eq:choosing_new_alpha_env}, events are chosen in a more complex way.
First, mimicking the  1st-stage filtering in the given run~$r$, the originally chosen $\alphae{r}{m}\in \envprotocol{m}$ is preventively purged of all byzantine events whenever they  would have caused more than $f$~agents to become faulty in~$r$. Note that,  in our transitional simulation of the adjusted run~$r'$, this is done prior to filtering~\eqref{eq:trans_filter_desir_env} by exploiting the  correctability of all agents.
Secondly, for all agents~$l$ outside the reliable causal cone at the current timestamp~$m$, i.e.,~with $(l,m) \in \silentmasses{r}{\theta} \sqcup \faultbuffer{r}{\theta}$, all events are removed, to comply with the total freeze among the silent masses~$\silentmasses{r}{\theta}$ and to make room for byzantine communication in the fault buffer~$\faultbuffer{r}{\theta}$. 
The resulting set complies with~$\envprotocol{}$ because all agents are delayable. For the silent masses, this is the desired result. 
For the fault buffer, on the other hand, byzantine sends are added for every correct or byzantine send in~$r$, thus, ensuring that the incoming information in the reliable causal cone in~$r'$ is the same as in~$r$. 
For the case when a faulty buffer node~$(l,m)$ sent no messages in the original run, $\failof{l}$ is added to make  the immediate future $(l,m+1)$ byzantine despite its silence, which is crucial for fulfilling Statement~\eqref{lem:multipede:bad} and simplifying bookkeeping for byzantine agents.\looseness=-1

The proof of~\eqref{eq:trans_filter_desir_env}--\eqref{eq:trans_filter_desir_ag} is by induction on $m =0,\dots,t-1$. 
To avoid overlong formulas, we abbreviate the right-hand side of~\eqref{eq:trans_filter_desir_env} by~$\Upsilon^m_\epsilon$ and the right-hand sides of~\eqref{eq:trans_filter_desir_ag}  for each $j \in \agents$ by~$\Xi_j^m$ for the specific~$\alphaag{j}{r'}{m}$~and~$ \alphae{r'}{m}$ defined in~\eqref{eq:choosing_new_alpha_ag}~and~\eqref{eq:choosing_new_alpha_env} respectively. 
Thus, it only remains  to show that $\betae{r'}{m}= \Upsilon^m_\epsilon$ and $(\forall j \in \agents)\,\betaag{j}{r'}{m} = \Xi_j^m$, or equivalently, further abbreviating $\Upsilon_j^m \ce  \Upsilon^m_\epsilon \cap \gevents[j]$,  that
\[
 \betae[j]{r'}{m} = \Upsilon_j^m
 \qquad\text{and}\qquad 
 \betaag{j}{r'}{m} = \Xi_j^m
 \]  for all $j \in \agents$, by simultaneous induction on~$m$. 
\looseness=-1

\emph{Induction step for the silent masses}  $(j,m) \in \silentmasses{r}{\theta}$. By~\eqref{eq:choosing_new_alpha_env},  $\alphae[j]{r'}{m}\ce \alphae{r'}{m} \cap \gevents[j] = \varnothing$, and filtering it yields $\Upsilon_j^m=\varnothing=\betae[j]{r'}{m}$ as prescribed by $\newfreezerule{}{}{}$.
 In particular, $go(j) \notin\betae[j]{r'}{m}$, thus, ensuring that filtering $\alphaag{j}{r'}{m}$, whatever it is, yields $\Xi_j^m=\varnothing = 
	\betaag{j}{r'}{m}$, once again in compliance with $\newfreezerule{}{}{}$ applied within $\silentmasses{r}{\theta}$.\looseness=-1

Before proceeding with the induction step for the remaining nodes, observe that events in $\alphae{r'}{m}$, if added to $r'(m)$, do not cross the byzantine-agent threshold~$f$, meaning that the 1st-stage filtering does not affect $\alphae{r'}{m}$:\looseness=-1
\begin{equation}
\label{eq:filter_futile}
\filtere[\leq f]{r'(m)}{\alphae{r'}{m},\alphaag{1}{r'}{m},\dots,\alphaag{n}{r'}{m}} = \alphae{r'}{m}.
\end{equation}
 Indeed there are two sources of byzantine events in  $\alphae{r'}{m}$: byzantine events from $\alphae{r}{m}$ that survived $\filtere[\leq f]{}{}$ in~\eqref{eq:choosing_new_alpha_env} and those pertaining to nodes in the fault buffer~$\faultbuffer{r}{\theta}$. 
The former were also present in~$\betae{r}{m}$ in the original run because the 2nd-stage filter $\filtere[B]{}{}$ only removes correct (receive) events. 
At the same time, for any $(l,m) \in \faultbuffer{r}{\theta}$, the immediate future $(l,m+1)$ was a faulty node in~$r$ by the definition of $\faultbuffer{r}{\theta}$. In either case, any agent  faulty in~$r'$ based on $\alphae{r'}{m}$ was also faulty by timestamp~\mbox{$m+1$} in~$r$. 
Additionally, any agent already  faulty  in $r'(m)$ was also faulty in $r(m)$ by Statement~\eqref{lem:multipede:nomorefaulty}. Since the number of agents faulty by~\mbox{$m+1$} in the original transitional run~$r$  could not exceed~$f$, adding $\alphae{r'}{m}$ to~$r'(m)$ does not exceed this threshold either. It follows from~\eqref{eq:filter_futile} that\looseness=-1
\begin{equation}
\label{eq:filter_only_B}
\Upsilon_{j}^m = \mathit{filter}^B_\epsilon(r'(m),\alphae{r'}{m},\alphaag{1}{r'}{m},\dots,\alphaag{n}{r'}{m}) \cap \gevents[j].
\end{equation}
\looseness=-1

\emph{Induction step for the fault buffer}  $(j,m) \in \faultbuffer{r}{\theta}$. For these nodes, the $\alphae[j]{r'}{m}$ part of  $\alphae{r'}{m}$ contains no correct events, hence, $ \filtere[B]{}{}$, which  only removes correct receives, has no effect. In other words,\looseness=-1
\begin{multline*}
\Upsilon_j^m\quad =\quad \alphae{r'}{m} \cap\gevents[j]\quad=\quad  \{\failof{j}\}\quad \sqcup\quad
\Bigl\{\fakeof{j}{\mistakefor{\gsend{j}{h}{\mu}{id}}{\tick}}  \,\Big|\\
\gsend{j}{h}{\mu}{id}\in \betaag{j}{r}{m}  \lor 
        (\exists A 
        )\,\fakeof{j}{\mistakefor{\gsend{j}{h}{\mu}{id}}{A}} \in \betae{r}{m}  \Bigr\}  \quad=\quad \betae[j]{r'}{m}
        \end{multline*}
   as    prescribed by $\newechorule{j}{m}{}$. As in the case of the silent masses,  $go(j)\notin\betae{r'}{m}$ guarantees that the $\Xi_j^m = \varnothing = \betaag{j}{r'}{m}$ requirement is fulfilled within~$\faultbuffer{r}{\theta}$.\looseness=-1
        
\emph{Induction step for the reliable causal cone}  $(j,m) \in \truecausalcone{r}{\theta}$. The case of the nodes with a reliable causal path to~$\theta$, whose immediate future remains correct in~$r$, is the final and also most complex induction step.  Recall that $\alphaag{j}{r'}{m} = \alphaag{j}{r}{m} \in P_j(r(m)) = P_j(r'(m))$ because within $\truecausalcone{r}{\theta}$ by Statement~\eqref{lem:multipede:same} $r'(m)=r(m)$. 
Thus, our choice of $\alphaag{j}{r'}{m}$ in~\eqref{eq:choosing_new_alpha_ag} is in compliance with transitionality. 
Since $(j,m+1)$ is correct, the $\alphae[j]{r}{m}\ce\alphae{r}{m}\cap \gevents[j]$ part of $\alphae{r}{m}$ contained no byzantine events and, hence, is unchanged by~\eqref{eq:choosing_new_alpha_env}. For the same reason it is not affected by 1st-stage filtering in either run. 
Thus,  the same set of $j$'s~events undergoes the 2nd-stage filtering in both the original run~$r$ and in our transitional simulation of the adjusted run~$r'$. 
Let us call this set of $j$'s~events~$\Omega_j$.\looseness=-1

Since both $\filtere[B]{}{}$ and $\newchatterrule{j}{m}{\havevoice{r}{\theta}}{}$ can only remove receive events, it immediately follows that  $\betae[j]{r'}{m}\subseteq \betae[j]{r}{m}$ and~$\Upsilon_j^m$  agree on all non-receive events.
Importantly, this includes $go(j)$ events, thus ensuring that $\Xi_j^m = \alphaag{j}{r}{m}$.\looseness=-1

A receive event $U=\grecv{j}{k}{\mu}{id} \in \Omega_j$ is retained in either run if{f} it is causally grounded by a matching send, correct or byzantine. 
Due to the uniqueness of GMI~$id$, as ensured by the injectivity of both~$id$ and $\mathit{global}$ functions, as well as Condition~(c) of the $t$-coherency of sets produced by $\envprotocol{}$, there is at most one agent~$k$'s node where such a matching send can originate from. 
If $id$~is not well-formed and no such send can exist, $U$~is filtered out from both $\betae{r}{m}$ and~$\Upsilon_\epsilon^m$, the former ensuring $U \notin \betae{r'}{m}$. 
The reasoning in the case such a node $\eta=(k,z)$ exists depends on where timestamp~$z$ is relative to~$m$ and where $\eta$~falls in our partition of nodes. 
Generally, to retain~$U$ in $\betae[j]{r}{m}$ and~$\Upsilon_j^m$, one must find  either  a correct send $V \ce \gsend{k}{j}{\mu}{id}$ or a faulty send $W_A \ce \fakeof{k}{\mistakefor{V}{A}}$ for some $A \in \gactions[k]\sqcup \{\tick\}$.\looseness=-1
\begin{compactitem}
\item  If $z > m$ is in the future of~$m$, then $U$~is filtered out from both $\betae{r}{m}$ and~$\Upsilon_\epsilon^m$, hence, $U \notin \betae{r'}{m}$.\looseness=-1
\item If $z\leq m$ and $\eta \in \silentmasses{r}{\theta}$, then, independently of filtering in~$r$, hap  $U \notin  \E\bigl(\newchatterrule{j}{m}{\havevoice{r}{\theta}}{r}\bigr) = \betae[j]{r'}{m}$ because the message's origin is outside the focus area. At the same time, no actions or events are scheduled at~$\eta$ in~$r'$ (for $z=m$ it follows from the already proven induction step for silent masses). Without either~$V$~or~$W_A$, event~$U$ is filtered out from $\Upsilon_\epsilon^m$.\looseness=-1
\item If $z < m$ and $\eta \in \faultbuffer{r}{\theta}$, then, by~\eqref{eq:echorule} in the definiton of $\newechorule{k}{z}{}$, only $W_\tick$~can save~$U$ in~$\Upsilon^m_j$ and  $W_{\tick}\in\betae[k]{r'}{z}$ if{f} either $V \in \betaag{k}{r}{z}$ or $W_A \in \betae[k]{r}{z}$ for some~$A$. Thus, filtering~$U$ yields the same result  in both runs, and $\newchatterrule{j}{m}{\havevoice{r}{\theta}}{}$ does not affect~$U$ because $\eta \in \havevoice{r}{\theta}$.\looseness=-1
\item If $z = m$ and $\eta \in \faultbuffer{r}{\theta}$, again only~$W_\tick$ can save~$U$ in~$\Upsilon^m_j$, this time by construction~\eqref{eq:choosing_new_alpha_env} of $\alphae{r'}{m}\notni go(k)$. 
Here $W_{\tick}\in\alphae[k]{r'}{m}$ if{f} either $V \in \betaag{k}{r}{z}$ or $W_A \in \betae[k]{r}{z}$ for some~$A$. Thus, filtering~$U$ yields the same result  in both runs, and $\newchatterrule{j}{m}{\havevoice{r}{\theta}}{}$ does not affect~$U$ because $\eta \in \havevoice{r}{\theta}$.\looseness=-1
\item If $z < m$ and $\eta \in \truecausalcone{r}{\theta}$, then $(k,z+1)$ is still correct in both~$r$~and~$r'$, hence,  no byzantine events such as~$W_A$ are present in either  $r(m)$ or $r'(m)$. Accordingly, only~$V$ can save~$U$ in this case. Since $\betaag{k}{r}{z}=\betaag{k}{r'}{z}$ by construction~\eqref{eq:chatterrule} of $ \newchatterrule{k}{z}{\havevoice{r}{\theta}}{}$, filtering~$U$ yields the same result  in both runs, and $\newchatterrule{j}{m}{\havevoice{r}{\theta}}{}$ does not affect~$U$ because $\eta \in \havevoice{r}{\theta}$.
\looseness=-1
\item If $z = m$ and $\eta \in \truecausalcone{r}{\theta}$, again $(k,m+1)$ is correct in $r$ meaning this time  that no  $W_A$ are present in~$\Omega_k$. Again, only $V$ can save $U$ from filtering. Since $\alphaag{k}{r}{m}=\alphaag{k}{r'}{m}$ by construction~\eqref{eq:choosing_new_alpha_ag} and the sets of events being filtered agree on $go(k) \in \Omega_k$, here too filtering $U$ yields the same result  in both runs, and $\newchatterrule{j}{m}{\havevoice{r}{\theta}}{}$ does not affect $U$  because $\eta \in \havevoice{r}{\theta}$.\looseness=-1
\end{compactitem}
This case analysis completes the induction step for the reliable causal cone, the induction proof, proof of Statement~\eqref{lem:multipede:trans}, and the proof of the whole Lemma~\ref{lem:sources_isolation}.\looseness=-1
\end{proof}

\section{Preconditions for Actions: Multipedes}
\label{sect:PrecActions}

Arguably the most important application of Lemma~\ref{lem:sources_isolation}, and,
hence, of causal cones, is  to derive preconditions for agents'
actions, cp.~\cite{ben2014beyond}. While  relatively simple in traditional settings, where events can be preconditions according to the
knowledge of preconditions principle~\cite{Mos15TARK}  and where Lamport's causal
cone  suffices, this is no longer true
in byzantine settings. As Theorem~\ref{lem:no-k-occurred} reveals, if $f>0$, an asynchronous agent can learn neither that it is (still) correct
  nor that a particular event\footnote{Actually, the reasoning in this section
also extends to actions, i.e., arbitrary haps.} really occurred.\looseness=-1

\begin{theorem}[\cite{FroCoS19}]\label{lem:no-k-occurred}
If $f\geq1$, then for any $o \in \events[i]$,  for any interpreted system $\I=(R^\chi,\pi)$ with
 any non-excluding agent-context $\chi=((\envprotocol{},\globalinitialstates,\tau_f,\Psi),P)$ where $i$~is gullible   and  every  $j\ne i$  is delayable,
 \looseness=-1
\begin{equation}  
\label{eq:lack_self_knowledge}
\I \models \lnot \K{i}{\trueoccurred{o}}
\qquad 
\text{and}\qquad
\I \models \lnot \K{i}{\correct{i}}
.
\end{equation}
\end{theorem}
These validities  can be shown by modeling the 
infamous  \emph{brain in a vat} scenario (see~\cite{FroCoS19} for details).

Theorem~\ref{lem:no-k-occurred} obviously implies that \emph{knowledge} of simple preconditions, e.g.,~events, is never achievable if byzantine agents are present. Settling for the next best thing, one could investigate whether $i$~knows~$o$ has happened relative to its own correctness, i.e.,~whether $K_i(\correct{i} \to  \trueoccurred{o})$ holds (cf.~\cite{MosSho93AI}), a kind of non-factive~\emph{belief} in~$o$. This means that $i$~can  be mistaken about~$o$ due to its own faults (in which case it cannot rely on any information anyway), not due to being misinformed by other agents. It is, however, sometimes overly restrictive to assume that $K_i(\correct{i} \to  \trueoccurred{o})$ holds in situations when $i$~is, in fact, faulty: typical specifications, e.g.,~for distributed agreement~\cite{lamport1982byzantine}, do not restrict the behavior of  faulty agents, and agents might sometimes learn that they are faulty. We therefore introduced the
\emph{hope} modality
\[
H_i \varphi \ce \correct{i} \to K_i(\correct{i} \to  \varphi),
\]
which was shown in~\cite{Fru19ESSLLI} to be axiomatized  by adding to~{\sf K45} the  axioms
$
\correct{i} \rightarrow (H_{i}\varphi \rightarrow \varphi)$, and
$
\neg \correct{i}\rightarrow H_{i}\varphi$, and
$
H_{i}\correct{i}$.

The following Theorem~\ref{th:byz_cone} shows that hope is also closely connected
to reliable causal cones, in the sense that events an agent can hope for must
lie within the reliable causal cone.
\looseness=-1

\begin{theorem}
\label{th:byz_cone}
    For a non-excluding agent-context  $\chi=((\envprotocol{}, \globalinitialstates, \tau_f,\Psi),\joinprotocol{})$ such that all agents     are gullible, correctable, and delayable, for a correct node $\theta=(i,t)$, and for an event $o \in \events$, if  all occurrences of   $O\in \gtrueevents$ such that  $\mathit{local}(O)=o$ happen outside the reliable causal cone~$\truecausalcone{r}{\theta}$ of a  run $r \in R^\chi$, i.e.,~if $
    O \in \betae{r}{m} \cap \gtrueevents[j] \,\,\&\,\, \mathit{local}(O)=o $ implies $(j,m) \notin\truecausalcone{r}{\theta}$, then  for any  $\I=(R^\chi,\pi)$,
      \[
    (\I,r,t) \nvDash H_i \trueoccurred{o}.
    \] 
        \end{theorem}
    \begin{proof}
    Constructing the first $t$~rounds  according to the adjustment from Lemma~\ref{lem:sources_isolation} and extending this prefix to an infinite run $r'\in R^\chi$ using the non-exclusiveness of~$\chi$, we obtain a run with no correct events recorded as~$o$. Indeed, in~$r'$, there are no events originating from~$\silentmasses{r}{\theta}$, no correct events from~$\faultbuffer{r}{\theta}$, and all events originating from~$\truecausalcone{r}{\theta}$, though correct, were also present in~$r$ and, hence, do not produce~$o$ in  local histories. At the same time, $r_i(t)=r'_i(t)$ by Lemma~\ref{lem:sources_isolation}\eqref{lem:multipede:same:i}, making $(\I,r',t)$ indistinguishable for~$i$, and $(\I,r',t)\vDash \correct{i}$  by Lemma~\ref{lem:sources_isolation}\eqref{lem:multipede:nomorefaulty}. \looseness=-1   \end{proof}

It is interesting to compare the results and proofs of Theorems~\ref{lem:no-k-occurred}~and~\ref{th:byz_cone}. Essentially, in the run~$r'$ modeling  the brain in a vat in the former, $i$~is a faulty agent that perceives events while none really happen. Therefore, $\K{i}{\trueoccurred{o}}$ can \emph{never} 
be attained. In the run~$r'$ constructed by Lemma~\ref{lem:sources_isolation}
in Theorem~\ref{th:byz_cone}, on the other hand, $i$~remains correct. The reason that
$H_i \trueoccurred{o}$ fails here is  that $o$~does not occur within the
reliable causal cone.
\looseness=-1 

Theorem~\ref{th:byz_cone} shows that, in order to act based on the hope that an event occurred, it is necessary that the event originates  from the reliable causal cone. 
Unfortunately,  this is not sufficient. Consider the case of a run~$r$ where no agent exhibits a fault: 
every causal message chain is reliable and the ordinary and reliable causal cones coincide. 
However, since up to $f$~agents \emph{could be} byzantine, it is trivial to modify~$r$ by seeding $\failof{j}$ events in round \textonehalf{} for several agents~$j$ in a way that is indistinguishable for  agent~$i$ trying to hope for the occurrence of~$o$. This would enlarge the fault buffer and shrink the reliable causal cone in the so-constructed adjusted run~$\hat{r}$. 
Obviously, by making different sets of agents byzantine (without violating~$f$, of course), one can fabricate multiple adjusted runs where $\hat{r}_i(t)=r_i(t)$ is exactly the same but fault buffers and reliable causal cones vary in size and shape. Any single one of those~$\hat{r}$ satisfying the conditions of Theorem~\ref{th:byz_cone}, in the sense that all occurrences of~$o$ happen outside its reliable causal cone, dash the \emph{hope} of~$i$ for~$o$ in~$r$. 

Thus, in order for~$i$ to have hope at~$(i,t)$ in run~$r$ that $o$~really occurred, it is necessary that some correct global version~$O$ of~$o$ (not necessarily the same one) is present somewhere (not necessarily at the same node) in the reliable causal cone of \emph{every} run~$\hat{r}$ that ensures $r_i(t)=\hat{r}_i(t)$. This gives rise to the definition of a \emph{multipede}, which ensures
$(\I,r,t) \models H_i \trueoccurred{o}$ according to Theorem~\ref{th:byz_cone}:

\begin{definition}[Multipede]\label{def:multipede} We say that a run~$r$ in a non-excluding
 agent context $\chi=((\envprotocol{}, \globalinitialstates, \tau_f,\Psi),\joinprotocol{})$ contains a $\multipede{o}{\theta}$ for event $o \in \events$ at some node $\theta=(i,t)$ iff, for all  runs~$\hat{r}\in R^\chi$ with $r_i(t)=\hat{r}_i(t)$, it holds that $o$~happens inside its reliable causal cone, i.e.,~that  
 \[(\exists (j,m) \in\truecausalcone{\hat{r}}{\theta})\, (\exists O \in \gtrueevents[j]) \bigl(O \in \beta_{\epsilon}^m(\hat{r})  \,\,\&\,\, \mathit{local}(O)=o\bigr).
 \]
\end{definition}

We obtain the following necessary condition for the existence of a multipede:\looseness=-1
\begin{theorem}[Necessary condition for a multipede]\label{thm:necessarymultipede}
    Given an arbitrary non-excluding agent-context  $\chi=\bigl((\envprotocol{}, \globalinitialstates, \tau_f,\Psi),\joinprotocol{}\bigr)$ such that all agents     are gullible, correctable, and delayable and for any run $r \in R^\chi$ in any interpreted system $\I=(R^\chi,\pi)$, if    \looseness=-1
    $
    (\I,r,t) \vDash H_i \trueoccurred{o} 
    $
    for a correct node $\theta=(i,t)$, i.e., if there is a $\multipede{o}{\theta}$ in $r$, then the following must hold:  Let $\mathit{Byz}^r_\theta \ce \{j \in \agents \mid (\exists m)(j,m) \in \faultbuffer{r}{\theta}\}$. For any $S \subseteq \agents \setminus(\{i\}\sqcup\mathit{Byz}^r_\theta)$ such that $|S|= f -|\mathit{Byz}^r_\theta|$, 
    there must exist a  witness $w_S\in\agents$ of some correct event $O_S \in \betae{r}{m_S} \cap \gtrueevents[w_S]$ such that  $\mathit{local}(O_S)=o$  and such that there is  causal path $\pastpath[\xi_S]{r}{(w_S,m_S)}{\theta}$ that does not involve agents from $S \sqcup \mathit{Byz}^r_\theta$. 
\looseness=-1
    \end{theorem}
    \begin{proof}
    Since, by Lemma~\ref{lem:sources_isolation}, the adjusted run $r'\in R^\chi$ and since the only faults up to~$t$ occur in~$r'$ in the fault buffer~$\faultbuffer{r}{\theta}$, i.e.,~pertain to agents from  $\mathit{Byz}^r_\theta$, for any~$S$ described above, one can construct first $t$~rounds  by setting $\beta^0_\epsilon(r^S) \ce  \beta^0_\epsilon(r') \sqcup \{\failof{j}\mid j \in S \sqcup \mathit{Byz}^r_\theta\}$ and keeping the rest of $r'$~intact. These first $t$~rounds  can be extended to complete infinite  runs $r^S\in R^\chi$ indistinguishable for~$i$ at~$\theta$ from either~$r'$~or~$r$ because the addition of $\failof{j}$ is imperceptible for agents and does not affect protocols. The only potentially affected element could have been $\filtere{}{}$ in the part ensuring byzantine agents do not exceed~$f$ in number, but it also behaves the same way as in~$r'$ because $|S| + |\mathit{Byz}^r_\theta|= f$.
Since $r_i(t)=r'_i(t)=r^S_i(t)$, we have $(\I,r^S,t) \vDash H_i \trueoccurred{o}$. Node~$\theta$ remains correct in these runs because $i \notin S$. 
Thus, by Theorem~\ref{th:byz_cone}, each run~$r^S$ must have a requisite correct event $O_S \in \beta^{m_S}_{\epsilon_{w_S}}(r^S)\cap \truecausalcone{r^S}{\theta}$. It remains to note that any such correct event from~$r_S$ must be present in~$r'$ and in~$r$ and any causal path in~$r^S$ exists already in~$r'$~and~$r$, according to the construction from Lemma~\ref{lem:sources_isolation}. 
Thus,  there must exist  a  causal path~$\xi$ in~$r$ from $(w_S,m_S)$  to~$\theta$ such that $\xi$~is reliable in~$r^S$. Finally, since all $f$~byzantine agents in~$r^S$, namely $S \sqcup \mathit{Byz}^r_\theta$, are made faulty from round~\textonehalf, path~$\xi_S$ being reliable in~$r^S$ means not involving these agents.\looseness=-1   \end{proof}

From the perspective of protocol design, arguably, of more interest are sufficient conditions for the existence of a multipede in a given
run. Whereas a sufficient condition could be obtained directly from 
Def.~\ref{def:multipede}, of course, identifying all the 
transitional runs~$\hat{r}$ with $r_i(t)=\hat{r}_i(t)$ is far 
from being computable in general. Actually, we conjecture that 
sufficient conditions cannot be formulated in a protocol-independent 
way at all. Unfortunately, however, protocol-dependence cannot be expected to
be simple either. For instance, even just varying the number 
and location of faults in~$r$ for suppressing $\trueoccurred{o}$ 
in a modified run~$\hat{r}$ could be non-trivial.
If $k$~agents are already faulty in run~$r$, at least $f-k$ ones can freely be used for this purpose. However, some of the $k$~byzantine faults in~$r$ may also be re-located in~$\hat{r}$, as agents that only become faulty after timestamp~$t$ cannot be part of any fault buffer. Rather than making them faulty, it would suffice to just freeze them. 
\looseness=-1

For instance, for the following communication structure with $f=2$ and agents~$1$~and~$2$ being byzantine (we omit the time dimension for simplicity's sake), \looseness=-1 
\[
\begin{tikzcd}[cells={nodes={draw=gray}}]
1 \arrow{r} & 2 \arrow{r} & 3 & 4 \arrow{l}
\end{tikzcd}
\]
they would both participate in the fault buffer, whereas, already 2~alone would suffice because even were 1~correct, the observed communication does not give it a chance to pass by~2. Depending on 1's~protocol, it might be possible to reassign~1 to the silent masses, thereby allowing to consider~4 as the second faulty agent and, thus, showing the impossibility for~3 to act in this situation.
An opposite outcome is possible in the following scenario:\looseness=-1
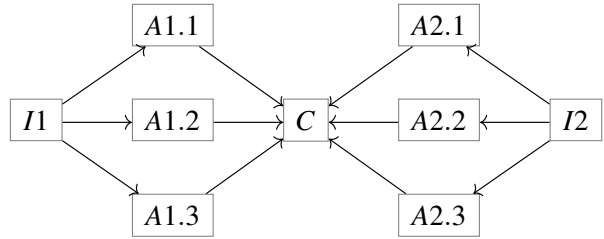
\[
\begin{tikzcd}[cells={nodes={draw=gray}}]
& A1.1 \arrow{rd} & & A2.1\arrow{ld}
\\
I1 \arrow{r}\arrow{ru}\arrow{rd} & A1.2 \arrow{r} & C &  A2.2\arrow{l} & I2\arrow{l}\arrow{lu}\arrow{ld}
\\
& A1.3\arrow{ru} & & A2.3\arrow{lu}
\end{tikzcd}
\]
Let $f=2$ and the faulty agents be~$A2.1$~and~$A1.1$. While the sufficient condition forces~$C$ to consider the case of both~$I1$~and~$I2$ being compromised and information originating from them unreliable, our necessary condition does not rule out $C$'s~ability to make a decision. Indeed, suppose~$I1$~and~$I2$ are investigators sending in their reports via three aides each. Having received 4~identical reports that are correct from~$A1.2$, $A1.3$, $A2.2$,~and~$A2.3$ and only 2~fake reports from~$A1.1$~and~$A2.1$, agent~$C$ would have been able to choose the correct version if the possibility of both investigators being compromised were off the table. Our method of adjusting the run does not allow us to move the faulty agent from~$A1.1$ to~$I1$ because it is not clear how $A1.1$~would have behaved were it correct and had it received a fake report from~$I1$. By designing a protocol in such a way that $A1.1$'s~correct behavior in such a hypothetical situation is different, we can eliminate the possibility of investigators being compromised and, thus, resolve the situation for~$C$.
\looseness=-1 

\section{Conclusions}
\label{sect:conclusions}

The main contribution of this paper is the characterization of the analog of Lamport's causal cone in asynchronous multi-agent systems with byzantine faulty agents. Relying on our novel byzantine runs-and-systems framework, we provided an accurate epistemic characterization of causality and the induced reliable causal cone in the presence of asynchronous byzantine agents. Despite the quite natural final shape of a reliable causal cone, it does not lead to simple conditions for ascertaining preconditions: the detection of what we called a multipede is considerably more complex than the verification of the existence of one causal path in the fault-free case. Since the agents' actions depend on the shape of multiple alternative reliable causal cones in byzantine fault-tolerant protocols like~\cite{ST87}, however, there is no alternative but to detect multipedes.
\looseness=-1 

Developing practical sufficient conditions for the existence of a multipede poses exciting challenges, which are currently being addressed in the context of the epistemic analysis of some real byzantine fault-tolerant protocols. This context-dependency is unavoidable, since the agent that tries to detect a multipede in a run lacks global information such as the actual members of the fault buffer.
On the other hand, the gap between the necessary and sufficient conditions can potentially be minimized by designing protocols based on the insights into the causality structure we have uncovered. For instance, while we treated all error-creating nodes as part of the fault buffer in our necessary conditions for a multipede, it is sometimes possible to relegate redundant parts of it into the silent masses. As this would allow to re-locate byzantine faults for intercepting more causal paths, one may design protocols in a way that does not allow this.\looseness=-1

A larger and more long-term goal is to extend our study to syncausality and the reliable
syncausal cone in the context of synchronous byzantine fault-tolerant multi-agent systems, and to possibly incorporate protocols explicitly into the logic. \looseness=-1
\medskip

\textbf{Acknowledgments.} We are grateful to Yoram Moses, Hans van Ditmarsch, and Moshe Vardi for fruitful discussions and valuable suggestions that have helped shape the final version of this paper. We also thank the anonymous reviewers for their comments and suggestions.

\bibliographystyle{eptcs}
\bibliography{shortbibdoi}
\appendix
\section*{Appendix}
\label{app:app}
\renewcommand{\thetheorem}{A.\arabic{theorem}}

\subsection*{Filter functions}
\begin{definition}
\label{def:filter_env}
The filtering function $\filtere{}{}{}$ for asynchronous agents with at most $f\geq 0$ byzantine faults is defined as follows. 

First, we define a subfilter 
$
\filtere[B]{}{}{}{} \colon \globalstates \times \wp(\gevents) \times \prod_{i=1}^n\wp(\gactions[i]) \to \wp(\gevents)
$
 that removes impossible receives: for a global state $h\in\globalstates$, set  $X_\epsilon \subseteq \gevents$, and sets $X_i \subseteq \gactions[i]$,
   \begin{multline*}
    \filtere[B]{h}{X_\epsilon, X_1, \dots, X_n}  \ce 
    X_\epsilon \setminus 
    \Bigl\{
    \grecv{j}{i}{\mu}{id}  \quad\Big|\quad 
    \gsend{i}{j}{\mu}{id} \notin h_\epsilon \quad\land\quad \\
    (\forall A \in \{\tick\}\sqcup\gactions[i])\,\fakeof{i}{\mistakefor{\gsend{i}{j}{\mu}{id}}{A}} \notin h_\epsilon \land
    (\gsend{i}{j}{\mu}{id} \notin X_i \lor go(i) \notin X_\epsilon) \land 
    \\ 
    (\forall A \in \{\tick\}\sqcup\gactions[i])\,\fakeof{i}{\mistakefor{\gsend{i}{j}{\mu}{id}}{A}} \notin X_\epsilon
    \Bigr\} ,
    \end{multline*}
    where $h_\epsilon$ is the environment's record of all haps in the global state $h$ and $O\in h_\epsilon$ ($O \notin h_\epsilon$) states that  the hap $O \in \ghaps$ is (isn't) present in this record of all past rounds, $X_\epsilon$ represents all events attempted by the environment and $X_i$'s represent all actions attempted by agents $i$ in the current round. 

Second, using  $X^B_{\epsilon_i} \ce X_\epsilon \cap \bigl(\bevents[i] \sqcup\{\sleep{i},\hib{i}\}\bigr)
$ and defining $\agents(\failed{h}{})$ to be the set of agents who have already exhibited  faulty behavior in the global state $h$, we define a subfilter
$
\filtere[\leq f]{}{}{}{} \colon  \globalstates \times \wp(\gevents) \times \prod_{i=1}^n\wp(\gactions[i]) \to \wp(\gevents)
$
 that removes all byzantine events in the situation when having them would have exceeded the $f$ threshold:
 \[
\filtere[\leq f]{h}{\,X_\epsilon,\, X_1,\, \dots,\, X_n}
\ce 
\begin{cases}
X_\epsilon & \text{if $\Bigl|\agents(\failed{h}{})\cup\left\{i \mid X^B_{\epsilon_i} \ne \varnothing\right\}\Bigr|\leq f$},
\\
X_\epsilon \setminus \bigsqcup\limits_{i\in \agents}X^B_{\epsilon_i} & \text{otherwise}.
\end{cases}
\]

The filter 
$
\filtere{}{}{} \colon  \globalstates \times \wp(\gevents) \times \prod_{i=1}^n\wp(\gactions[i]) \to \wp(\gevents)
$ 
is obtained by composing these two subfilters, with the $\leq f$ subfilter applied first:
\[
\filtere{h}{\,X_\epsilon,\, X_1,\, \dots,\, X_n} \ce
\filtere[B]{h}{\filtere[\leq f]{h}{\,X_\epsilon,\, X_1,\, \dots,\, X_n},\, X_1,\, \dots,\, X_n } .
\]
\end{definition}

The composition in the opposite order could violate causality if a message receipt is preserved by $\filtere[B]{}{}$ based on a byzantine send in the same round, which is later removed by $\filtere[\leq f]{}{}$.

\begin{definition}
\label{def:filter_ag}
The filters 
$
\filterag{i}{}{} \colon  \prod_{j=1}^n\wp(\gactions[j]) \times \wp(\gevents) \to \wp(\gactions[i])
$ 
for agents' actions are defined as follows: for $X_\epsilon$ representing all environment's events  and $X_i$ representing all actions attempted by agent $i$ in the current round,
    \begin{equation*}
  \filterag{i}{X_1,\dots,X_n}{X_\epsilon}\ce\begin{cases}
    X_i & \text{if } go(i) \in X_\epsilon,\\
        \varnothing & \text{otherwise}.
    \end{cases}
    \end{equation*}
\end{definition}

\subsection*{Update functions}

Before defining the update functions, we need several auxiliary functions:
\begin{definition}
We use  a function
 $
 \mathit{local} \colon  \gtruethings \to  \haps
 $
 converting \emph{correct} haps from the global format into the local formats for the respective agents in such a way that, for any $i,j \in \agents$, any $t \in \bbbt$,  any $a \in \actions[i]$, any $\mu \in \msgs$, and any $M \in \bbbn$:
 \begin{center}
 \begin{tabular}{rl@{\qquad\qquad}rl}
 1. & $\mathit{local}(\gtrueactions[i]) = \actions[i]$; & 3.  &
 $\mathit{local}\bigl(\glob{i,t}{a}\bigr) = a$;
 \\
2. &  $\mathit{local}(\gtrueevents[i]) = \events[i]$; &  4. &
 $\mathit{local}{\bigl(\grecv{i}{j}{\mu}{M}\bigr)} = \recv{j}{\mu}$.
 \end{tabular}
 \end{center}
For all other haps, the localization cannot be done on a hap-by-hap basis because system events and byzantine events
$\fakeof{i}{\mistakefor{A}{\tick}}$  do not create a local record. Accordingly,
    we define a \emph{localization function}
    $
    \sigmaof{} \colon \wp(\gthings)  \to  \wp(\truethings)
    $ 
    as follows: for each $X \subseteq \gthings$,  
    \begin{multline*}
    \sigmaof{X} \ce \mathit{local}\Bigl(\mathstrut^{\mathstrut}(X \cap \gtruethings) \quad\cup \\
 \{E \in \gtrueevents \mid (\exists i)\,\fakeof{i}{E}\in X\}  \cup 
 \{A' \in \gtrueactions\mid (\exists i)(\exists A)\,\fakeof{i}{\mistakefor{A}{A'}}\in X\} \Bigr).
 \end{multline*}
\end{definition}

\begin{definition}
\label{def:state-update}
        \index{${\updatee{}{},\updateag{i}{}{}{}}$}
        We abbreviate 
        $X_{\epsilon_i} \ce X_{\epsilon} \cap \gevents[i]$ for  performed events $X_\epsilon \subseteq \gevents$ and actions $X_i \subseteq \gtrueactions[i]$ for each $i \in \agents$.            Given a global state  
           $
           \run{}{t} = \bigl(\run{\epsilon}{t}, \run{1}{t}, \dots, \run{n}{t}\bigr) \in \globalstates
           $, we define 
               agent~$i$'s
               $
               \updateag{i}{}{}{} \colon \localstates{i} \times \wp(\gtrueactions[i])\times\wp(\gevents) \to \localstates{i} 
               $  
            that   outputs a new local state from~$\localstates{i}$ based on $i$'s actions~$X_i$ and events~$X_{\epsilon}$:
               \[
                   \updateag{i}{\run{i}{t}}{X_i}{X_\epsilon} \ce 
                   \begin{cases}
                   \run{i}{t} & 
                       \text{if $\sigma(X_{\epsilon_i})=\varnothing$ and $X_{\epsilon_i} \cap \{go(i), sleep(i)\} = \varnothing$}, \\
                       \left[\sigmaof{ X_{\epsilon_i}\sqcup X_i} \right] : \run{i}{t} & \text{ otherwise }
                   \end{cases}
               \]
               (note that in transitional runs, $\updateag{i}{}{}{}{}$ is always used after the  action  $\filterag{i}{}{}$, thus, in the absence of $go(i)$, it is always the case that  $X_i = \varnothing$).
   
        Similarly, the \emph{environment's state}
        $
        \updatee{}{} \colon \localstates{\epsilon} \times \wp\left(\gevents\right) \times \prod_{i=1}^n\wp(\gactions[i])  \to \localstates{\epsilon}
        $
         outputs a new state of the environment based on events  $X_{\epsilon}$ and all actions $X_i$:      
        \[
        \updatee{\run{\epsilon}{t}}{X_{\epsilon},X_1,\dots, X_n} \ce (X_{\epsilon} \sqcup X_1 \sqcup \dots \sqcup X_n) \colon \run{\epsilon}{t} 
        .\] 
        Accordingly, the global update function
                $
        \update{}{} \colon \globalstates \times \wp\left(\gevents\right) \times \prod_{i=1}^n\wp(\gactions[j])  \to \globalstates
        $
modifies
         the global state  as follows:
        \begin{multline*}
        \update{\run{}{t}}{X_{\epsilon},X_1,\dots, X_n} \ce \Bigl( \updatee{\run{\epsilon}{t}}{X_{\epsilon},X_1,\dots, X_n}, \\\updateag{1}{\run{1}{t}}{X_1}{X_\epsilon}, \dots, \updateag{n}{\run{n}{t}}{X_n}{X_\epsilon} \Bigr)
        .
        \end{multline*}
\end{definition}

\end{document}